\documentclass[runningheads]{llncs}
\usepackage[T1]{fontenc}
\usepackage{color}

\usepackage{orcidlink}
\usepackage{paralist}
\usepackage{xcolor}
\usepackage{listings}
\usepackage{quantikz}
\usepackage{array}
\usepackage{amsfonts} %provides \mathbb
\usepackage{xfrac} % provides \sfrac
\usepackage{amssymb} % provides \triangleq
\usepackage{xspace} % provides \xspace
\usepackage{stmaryrd} % provides \llbracket and \rrbracket
\RequirePackage{thm-restate} %provides restatable environment
\usepackage{hyperref} % provides \hypertarget
\usepackage{proof} % provides \infer
\usepackage[]{algorithm2e}
\usepackage{tabularx}
\usepackage{booktabs} % provides \toprule
\usepackage{float}
\usepackage{wrapfig}
\usepackage{placeins} % provides \FloatBarrier
\usepackage{quantikz}
\usepackage[english]{babel}

\definecolor{cat1}{RGB}{220, 220, 220} % Light gray
\definecolor{cat2}{RGB}{180, 180, 180} % Medium-light gray
\definecolor{cat3}{RGB}{140, 140, 140} % Medium-dark gray
\definecolor{cat4}{RGB}{100, 100, 100} % Dark gray
\usepackage[disable]{todonotes}

\newcommand{\ms}[1]{\todo[color=red!40]{MS: #1}}

\usepackage{tikz}
\usetikzlibrary{shapes, positioning, arrows.meta, backgrounds}
\usetikzlibrary{positioning, shapes, arrows.meta, fit, backgrounds}
\usepackage{pwhile}
\usepackage{paper}

\begin{document}
\title{Automated Expected Cost Analysis for Quantum Programs}
%
% \titlerunning{AutoCostQ}  % abbreviated title for the running head
%
\author{Georg Moser\inst{1}\raisebox{1ex}{\orcidlink{0000-0001-9240-6128}} \and
Michael Schaper\inst{1}\raisebox{1ex}{\orcidlink{0009-0004-6620-5073}}}
\authorrunning{G. Moser and M. Schaper}
\institute{
 Universität Innsbruck, Technikerstr. 21A, 6020 Innsbruck, Austria
\email{\{georg.moser,michael.schaper\}@uibk.ac.at}}
%

% \author{Anonymous Authors}
% \institute{Anonymous Institution}

\maketitle              % typeset the header of the contribution

\begin{abstract}
    % Quantum computing is advancing rapidly, but programming remains challenging and error-prone, highlighting the need for automated tools. 
    % We present $\TOOL$, a fully automated static program analysis tool for expected costs of mixed classical-quantum programs.
    % The tool is based on a recently proposed quantum expectation transformer framework, inspired by Dijkstra's predicate transformer and Hoare logic.
    % It supports programs with advanced features like mid-circuit measurements and classical control flow.
    % We evaluate our tool through a number of case studies taken from the literature and infer precise cost bounds fully automatically that previously required manual derivation.\gm{expand}\ms{?}
  
    In recent years, quantum computing has gained a substantial amount of momentum, and the capabilities of
    quantum devices are continually expanding and improving. Nevertheless, writing a quantum program from scratch remains
    tedious and error-prone work, showcasing the clear demand for automated tool support.
    We present $\TOOL$, a \emph{fully automated} static program analysis tool that yields a precise
    expected cost analysis of \emph{mixed classical-quantum programs}. $\TOOL$ supports programs with advanced features
    like mid-circuit measurements and classical control flow. 
    The methodology of our prototype implementation is based on
    a recently proposed quantum expectation transformer framework, generalising Dijkstra's predicate transformer and Hoare logic.
    The prototype implementation~$\TOOL$ is evaluated on a number of case studies taken from the literature and online references.
    $\TOOL$ is able to fully automatically infer precise upper bounds on the expected costs that previously could only be derived by tedious
    manual calculations.
    
    % In this work, we work towards
    % this challenge by considering tooling support for the analysis of quantitative properties of mixed
    % classical-quantum programs.
    %
    % More precisely, we present a fully automated approach for the expected cost analysis of quantum programs. Our
    % automation rests upon a recently introduced quantum expectation transformer, a framework closely
    % related to Dijkstra's predicate transformer and Floyd-Hoare's logic.
    % To the best of our knowledge, the prototype implementation presented here is the first automation
    % of an expectation transformer framework for quantum programs.
    %
    % We evaluate our approach through a number of case studies taken from the literature.
    % These programs make use of advanced features including mid-circuit measurements and classical control flow.
    % For almost all case studies,
    % this showcases how our method can fully automatically derive a precise and optimal
    % expected cost analysis that so far could only be obtained manually.
\end{abstract}

\section{Introduction}

Quantum computing (QC) has advanced rapidly, with competing hardware technologies (ion traps, superconducting circuits, neutral atoms, optics, \dots) and diverse computational models (MBQC~\cite{RBB03}, QRAM~\cite{Knill96}, ZX-diagrams~\cite{CD11},  quantum circuits, \ldots) driving progress, cf.~\cite{Greengard25c}.
At the software level, various programming languages and tools have emerged~\cite{BBH:2021,BandicFA22}, yet writing quantum programs remains challenging and error-prone. 
Programmers must master quantum mechanics, algorithms, and programming languages, often without adequate debugging or verification tools, leading to frequent bugs, cf.~\cite{Paltenghi:2022}. 
Consequently, developers face a steep learning curve and limited support for testing, debugging, or automated reasoning.
Automated program analysis for high-level quantum programs is critically lacking, especially for dynamic algorithms, which involve mid-circuit measurements, classical control flow and repetition.

In this context, it is important to distinguish between static analysis of functional properties and non-functional properties.
Further, it is also useful to distinguish between (static) quantum circuits and (dynamic) quantum programs with mid-circuit measurements and classical control flow.
Various tools provide verification of functional properties.
For example, \textsf{SymQV}~\cite{BauerMarquartLS23} and \textsf{Silver}~\cite{LewisZS24} focus on quantum circuits, while \textsf{AutoQ~2.0}~\cite{ChenCHHLLT25} extends to dynamic quantum programs.
However, one of the primary motivations for QC is the potential for computational speed-up, making it essential to reason about non-functional properties including performance and cost of execution.
Existing resource estimation tools such as Microsoft's \textsf{Azure Quantum Resource Estimator}~\cite{AQRE} and Google's \textsf{Qualtran}~\cite{qualtran} provide hardware-specific cost estimates but are limited to static quantum circuits.
They cannot handle dynamic algorithms that leverage mid-circuit measurements.
Examples of such algorithms include \emph{repeat-until-success} circuits~\cite{PaetznickS14}, which use mid-circuit measurements with classical feedback, and \emph{weak Grover search}~\cite{MartinezH22}, which combines quantum amplitude amplification with classical control.
The former was previously featured as a showcase on the IBM Quantum\textsuperscript{\textregistered{}} Learning platform.\footnote{\url{https://learning.quantum.ibm.com/tutorial/repeat-until-success}, last accessed January 2026.}
This motivates the development of an automated (expected) cost analysis for dynamic quantum algorithms.

% Our work rises to this challenge and introduces the tool~\TOOL\ for \emph{expected cost analysis}, leveraging recent theoretical advances and providing practical automation for real-world quantum programs. 
% In particular, the tool \emph{fully automatically} infers sound upper bounds on the expected stopping time of real-world quantum programs.
% More generally, program costs can be modelled via a \emph{ticking} instrumentation.
% seamlessly representing a non-monotonic cost model.

\paragraph{Contributions.}
The main contribution of this paper is to detail the fundamentals and technical realisation of
a tool---dubbed~\TOOL---that can automatically infer expected costs of quantum programs written in the \IMQ\ programming language, a mixed classical-quantum imperative language with support for quantum operations and measurements, cf.~\cite{AMPPZ:LICS:22}.
More precisely, we make the following contributions:
%
% GM: SPACE change to inparaenum if 
\begin{itemize}
%\begin{inparaenum}[1)]  
\item
  A novel term-based representation of the quantum expectation transformer by Avanzini et al.~\cite{AMPPZ:LICS:22,AMPP24}.
This representation suitably abstracts the underlying semantic model for cost expectations and replaces fixed-point constructions via first-order templates.
\item
  % The implementation of the tool \TOOL.
  A prototype implementation of the introduced term-based representation of the quantum expectation transformer
  in the tool~\TOOL. 
  $\TOOL$ fully automatically derives sound upper bounds on the cost of the input program in
  expectation. The prototype implementation will be submitted to the artefact evaluation.
\item
  Finally, an extensive validation of the tool on various quantum algorithms from the literature and online resources.
\end{itemize}
%\end{inparaenum}

\paragraph{Outline.}

Section~\ref{s:overview} provides a high-level overview of our contributions, and Section~\ref{s:related} discusses related work.
Section~\ref{s:preliminaries} describes syntax and semantics of the \IMQ\ programming language and recalls the expectation transformer.
Section~\ref{s:automation} details the implementation of the tool \TOOL\ and validates the tool on various quantum algorithms from the literature.
Finally, we conclude in Section~\ref{s:conclusion}.
The Appendix in Section~\ref{s:appendix} contains supplementary information, additional code examples and omitted proofs.

\clearpage

\section{Overview}
\label{s:overview}

\newcommand{\cq}{\Variable{ctrl}}
\newcommand{\trgt}{\Variable{trgt}}
\newcommand{\notdone}{\Variable{repeat}}

\newcommand{\xtwo}{-X}

\begin{wrapfigure}{r}{0.33\textwidth}
\vspace{-\baselineskip}
%\vspace{0mm}
\centering
\begin{minipage}[t]{0.32\textwidth}
    $
    {-}X(\x^{\Bool},\q_1^{\Qubits}, \q_2^{\Qubits}) \quad\triangleq\quad \\
    \phantom{xx}\pw{
        \tikzrom{cq-mark1}
      \q_1 <* H; \\
      \q_1,q_2 <* CNOT; \\
      \x <- \MEAS{\q_1};\hspace{1ex}\tikzrom{cq-mark3}
      \tikzrom{cq-mark2} \\
      \WHILE \neg \x \DO {
        \CONSUME~1; \\
        \q_1 <* X; \\
        \tikzrom{cq-mark1}
        \q_1 <* H; \\
        \q_1,q_2 <* CNOT; \hspace{1ex}\tikzrom{cq-mark3} \\
        \x <- \MEAS{\q_1}; 
      }
    }
    $
% \begin{quantikz}
%   % q1 *= X
%   % q2 *= H
%   % q1,q2 *= CNOT
%   % x = meas q1
%   \lstick{$\q_1$} & \gate{X} & \ctrl{1} & \meter{} \qw & \qw \\
%   \lstick{$\q_2$} & \gate{H} & \targ{}  & \qw                    & \qw \\
%   \lstick{$\x$}   & \cw      & \cw      & \cwbend{-2}                    & \cw
% \end{quantikz}
 \end{minipage}
\caption{The $-X$ program.}
\label{fig:X2}
\vspace{-\baselineskip}
\end{wrapfigure}%

This work presents \TOOL, a static program analysis tool for mixed classical-quantum programs, focusing on expected cost analysis. 
The input is a program written in the \IMQ\ language, an extension of Dijkstra's Guarded Command Language with support for quantum operations and cost annotations, as formally introduced in~\cite{AMPPZ:LICS:22,AMPP24}.
The output is either a sound upper bound (i.e. an arithmetic expression over input state variables) on the expected cost of executing the program or \texttt{Unknown} if the analysis fails to infer a bound.
It is important to note that the expected cost analysis of quantum programs is, in general, an undecidable problem~\cite{AMPP24}.

% \vspace{-\baselineskip}
% %\vspace{0mm}
% \centering
% \caption{Quantum circuit for the $-X$ program: the body applies $H$, $CNOT$, measures $\q_1$ into $\x$, and, if $\x=1$, applies $X$ to $\q_1$ and repeats.}
% \label{fig:X2-circuit}
% \vspace{-\baselineskip}
% \end{wrapfigure}%

Figure~\ref{fig:X2} shows the $-X$ program adapted from~\cite{ChenCHHLLT25}, which operates on a two-qubit system ($\q_1$ and $\q_2$).
The program applies a non-standard $-X$ gate to the second qubit, which performs the \mbox{Pauli-$X$} (i.e. the quantum \texttt{Not} gate) and negates the amplitude.
For instance, it transforms the quantum state (in Dirac notation) $\alpha \ket{10} + \beta \ket{11}$ into $-\beta \ket{10} - \alpha \ket{11}$.

The program begins by applying the Hadamard gate $\oper{H}$ to $\q_1$, which creates a superposition by transforming basis states as $H\ket{0} = \frac{1}{\sqrt{2}}(\ket{0}+\ket{1})$ and $H\ket{1} = \frac{1}{\sqrt{2}}(\ket{0}-\ket{1})$.
Next, the controlled-NOT gate $\oper{CNOT}$ is applied with $\q_1$ as control and $\q_2$ as target, which flips the target qubit if and only if the control qubit is in state $\ket{1}$, creating entanglement between the qubits.
The qubit $\q_1$ is then measured.
For input state $\alpha\ket{00} + \beta\ket{01} + \gamma\ket{10} + \delta\ket{11}$ the measurement yields $\ket{0}$ with probability $\sfrac{1}{2}(|\alpha+\gamma|^2 + |\beta+\delta|^2)$ and $\ket{1}$ with probability $\sfrac{1}{2}(|\alpha-\gamma|^2 + |\beta-\delta|^2)$.
% This requires careful computations with complex amplitudes.
%
Measurements introduce probabilistic branching, reflecting the inherent uncertainty in quantum programming.
If the measurement result is $\ket{0}$ the program enters the while loop with quantum state $\alpha \ket{00} + \beta \ket{01}$ for some (normalised) amplitudes $\alpha$ and $\beta$.
Then gate $\oper{X}$ is applied to reset $\q_1$, and the sequence of $\oper{H}$, $\oper{CNOT}$, and measurement is repeated.
This loop continues until the measurement yields $\ket{1}$, at which point the desired $-X$ transformation has been successfully applied to $\q_2$.
Each iteration consumes one unit of cost (via $\CONSUME~1$).
%
% Due to the representation of quantum states and the interplay between quantum and classical components, static analysis of quantum programs is challenging.
% Consider as input state $\alpha\ket{00} + \beta\ket{01} + \gamma\ket{10} + \delta\ket{11}$.
% Then, the first measurement of $\q_1$ yields $\ket{0}$ with probability $\sfrac{1}{2}(|\alpha+\gamma|^2 + |\beta+\delta|^2)$ and $\ket{1}$ with probability $\sfrac{1}{2}(|\alpha-\gamma|^2 + |\beta-\delta|^2)$.
% This requires careful computations with complex amplitudes and dynamic probabilities.
%
%The tool~\texttt{AutoQ 2.0}~\cite{ChenCHHLLT25} provides a Hoare-style verification with user-provided predicates and loop invariants to verify the intended behaviour.
%
% In this work, the \texttt{tick} instruction serves as a cost annotation to track resource usage during program execution.
For the ${-}X$ program, the total number of ticks represents the total number of loop iterations executed.
Due to the probabilistic nature of QC, we focus on the \emph{expected} number of iterations.
\TOOL\ fully automatically infers the tight bound $2 \cdot (\sfrac{1}{2} + a_{13} + a_{24})$ on the expected cost of executing this program (here $a_{13}$ and $a_{24}$ are relevant real-valued expressions of the density matrix representation of the initial quantum state).
For the automated analysis of this program, \TOOL\ needs to reason about quantum operations, dynamic probabilistic branching due to measurements, and fixed point computation of loop semantics.
This requires a suitable abstract representation of quantum states and a sophisticated constraint generation and solving procedure.

\begin{figure}[t]
  \centering
\begin{tikzpicture}[
  box/.style={rectangle, draw, rounded corners, minimum height=2em, minimum width=4em, align=center, font=\small},
  arrow/.style={-Stealth, thick},
  decision/.style={diamond, draw, minimum height=2em, minimum width=4em, align=center, font=\small},
  process/.style={box, minimum width=6em},
  inputoutput/.style={box, trapezium, trapezium left angle=70, trapezium right angle=110, minimum width=4em},
  ]

  \node[inputoutput] (input) at (0,0) {Quantum\\Program};
  \node[process, right=1.5cm of input] (symbolic) {\textbf{Symbolic}\\\textbf{Transformer}};
  \node[process, below=0.7cm of symbolic] (term) {Term\\Constraints};
  \node[process, right=1.5cm of term] (cost) {Cost\\Constraints};
  \node[process, below=0.7cm of cost] (poly) {Polynomial\\Constraints};
  \node[process, left=1.5cm of poly] (smt) {Certificate\\Constraints};
  \node[process, left=1.5cm of smt] (solver) {SMT Solver};
  \node[inputoutput, right=4.2cm of symbolic] (output) {Cost or\\\texttt{Unknown}};
  \node[left=0.2cm of symbolic](A){};
  \node[right=0.1cm of cost](B){};

  \node[draw, dashed, fit=(A)(symbolic)(cost)(poly)(smt)(B)] (tool) {};
  \draw[arrow] (input) -- (symbolic);
  \draw[arrow] (symbolic) -- (output);
  % \draw[arrow] (symbolic) -- (term);
  \draw[arrow] (symbolic) -- node[right, font=\small, align=center, text=gray] {Sec.~\ref{ss:inference-term-constraints}} (term);
  % \draw[arrow] (term) -- node[above, font=\small, align=center] {Polynomial\\Template} (cost);
  \draw[arrow] (term) -- node[above, font=\small, align=center, text=gray] {Sec.~\ref{ss:inference-cost-constraints}} (cost);
  % \draw[arrow] (cost) -- (poly);
  \draw[arrow] (cost) -- node[right, font=\small, align=center, text=gray] {Sec.~\ref{ss:inference-polynomial-constraints}} (poly);
  % \draw[arrow] (poly) -- node[above, font=\small, align=center] {Certificate\\Template} (smt);
  \draw[arrow] (poly) -- node[above, font=\small, align=center, text=gray] {Sec.~\ref{ss:inference-certificate-constraints}} (smt);
  \draw[arrow] (smt) -- (solver);
  \draw[arrow] (solver) -- (smt);
  \draw[arrow, bend left=40] (smt.north west) to (symbolic.south west);
\end{tikzpicture}
\caption{Overview of the expected cost analysis tool \TOOL.}
\label{fig:workflow}
\vspace{-\baselineskip}
\end{figure}
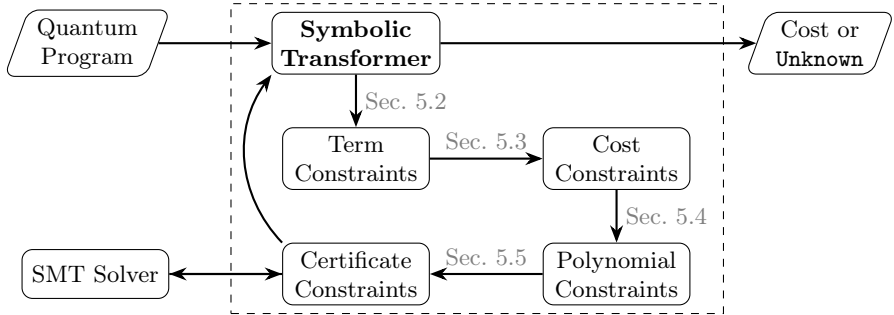

Figure~\ref{fig:workflow} provides an overview of the tool's architecture and workflow, highlighting its key components and their interactions.
The implementation automates the formal framework of quantum expectation transformers~\cite{AMPPZ:LICS:22,AMPP24} and is inspired by established techniques for expected cost analysis of probabilistic programs~\cite{AMS20,AMS23}.
%
% The formal framework provides a semantic model of the expected cost of quantum programs.
For automation, we introduce a \emph{symbolic transformer}: a syntactic variant of the quantum expectation transformer that operates on symbolic term-expectations (Subsection~\ref{ss:inference-term-constraints}).
To keep analysis tractable, the transformer manipulates a partial, algebraic representation of density matrices and probabilities.
Furthermore, fixed-point computations are avoided by generating upper invariant \emph{term constraints} for loops.
Term constraints generated during symbolic evaluation are refined to \emph{cost constraints} (Subsection~\ref{ss:inference-cost-constraints}), which have been successfully used for automation in classical and probabilistic settings~\cite{AMS20,AMS23,pubs}.
Then, cost constraints are further simplified into \emph{polynomial constraints} (Subsection~\ref{ss:inference-polynomial-constraints}).
Finally, polynomial constraints are reduced to \emph{certificate constraints} amenable for \textsf{SMT} solvers (Subsection~\ref{ss:inference-certificate-constraints}).
\ms{expanded this to clarify the loop}
The analysis algorithm is inherently iterative. 
The tool decomposes the program into strongly connected components and loops and employs a bottom-up inference strategy.
Inferred certificates are fed back to the symbolic transformer and propagated until a sound bound is inferred for the entire program.
If the polynomial constraints cannot be solved, the tool either reports \texttt{Unknown} or backtracks to recover from heuristic choices.

\paragraph{Motivating Example.}
\ms{added}
One prominent application of mid-circuit measurements are \emph{Repeat-Until-Success (RUS)} programs presented by Paetznick and Svore~\cite{PaetznickS14}. 
These algorithms are used to synthesise unitaries by repeated application of a fixed sequence of elementary quantum operations followed by a measurement on control qubits, until the desired measurement outcome is achieved. 
Notably, RUS programs maintain the quantum state of the qubits across iterations, rather than preparing always a fresh state.
% This synthesis approach often provides circuits with a smaller (expected)\pw{T}~gate count than other synthesis methods without mid-circuit measurements. Repeat-until-success algorithms are not meant to be stand-alone, but function as helper programs to be substituted into larger contexts as subprograms.

The repeat-until-success approach has been successfully realised on IBM Quantum\textsuperscript{\textregistered{}} devices, showcasing the practical application of dynamic circuits with mid-circuit measurements.%
\footnote{\url{https://learning.quantum.ibm.com/tutorial/repeat-until-success}, last accessed January 2026.}
% In what follows, we discuss the IBM showcase available online\footnote{\url{https://learning.quantum.ibm.com/tutorial/repeat-until-success}, last accessed January 2026.} in more detail.
% The source program is written using the \textsf{Qiskit SDK}~\cite{Wille:2019} and can be compiled to the \textsf{OpenQASM3} intermediate representation.
% It also can be executed on the IBM Quantum\textregistered{} platform.
%
%
The quantum circuit of the showcase is depicted in~Figure~\ref{fig:trial-circuit}.
% In contrast to the original specification in which all qubits are initialised to the basis state $\ket{0}$ by default, only the control qubits are explicitly set to $\ket{0}$, since no assumptions about the input state are made in our setting.
%
This quantum program has three qubits, two control qubits $\cq_1$ and $\cq_2$ and one target qubit $\trgt$.
The classical variables $\x_1$ and $\x_2$ are used to store and process the measurement outcomes of the control qubits, while the classical variable $\notdone$ is used to propagate the success of the operation.
On a successful measurement outcome of the control qubits---$\cq_1$ and $\cq_2$ are measured to be in state $\ket{00}$---%
the rotation operation $\pw{R_X}(\theta)$, with $\cos \theta = \sfrac{3}{5}$, has been successfully performed on the target qubit.

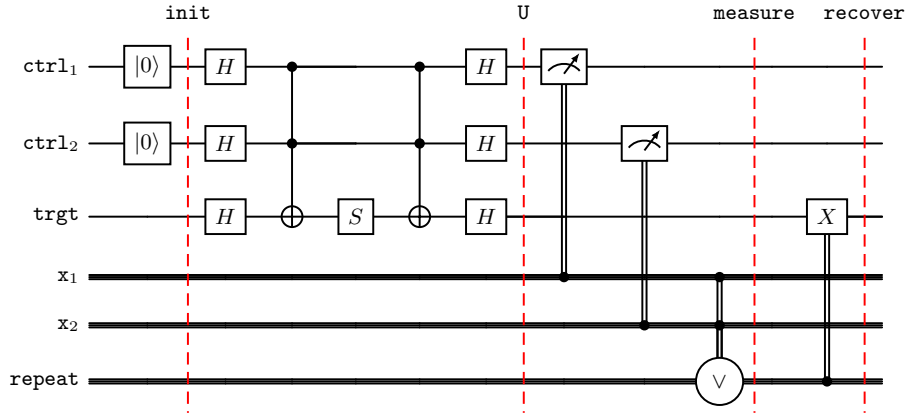
\begin{figure}[t]
\centering
\scalebox{0.93}{
\begin{quantikz}
\lstick{$\cq_1$}    & \gate{\ket{0}} \slice{\pw{init}}  & \gate{H} & \ctrl{1} & \qw      & \ctrl{1} & \gate{H} \slice{\pw{U}} & \meter{}    &         & \slice{\pw{measure}} &  & \slice{\pw{recover}}   & \\
\lstick{$\cq_2$}    & \gate{\ket{0}}   & \gate{H} & \ctrl{1} & \qw      & \ctrl{1} & \gate{H} &             & \meter{}    & &    & & \\
\lstick{$\trgt$}    &                  & \gate{H} & \targ{}  & \gate{S} & \targ{}  & \gate{H} & \qw         &             & & & \gate[wires=1]{X} & \qw \\
\lstick{$\x_1$}     & \cw              & \cw      & \cw      & \cw      & \cw      & \cw      & \cwbend{-3} & \cw         & \cwbend{2} & \cw &  \cw & \cw\\
  \lstick{$\x_2$}     & \cw              & \cw      & \cw      & \cw      & \cw      & \cw      & \cw         & \cwbend{-3} & \cwbend{1} & \cw & \cw & \cw \\
\lstick{$\notdone$} & \cw              & \cw      & \cw      & \cw      & \cw      & \cw      & \cw         & \cw         & \gate[style={shape=circle}]{\ensuremath{\vee}} \cw & \cw & \cwbend{-3} & \cw
\label{circuit:notdone}  
\end{quantikz}
}
\vspace{-2mm}
\caption{The quantum circuit for the $RUS$ program.}
\label{fig:trial-circuit}
\end{figure}

This circuit is executed repeatedly until a successful measurement outcome is achieved to perform the rotation on the target qubit.
\TOOL\ successfully infers the expected stopping time (i.e. the expected number of iterations until a successful measurement) of $\sfrac{8}{5}$ 
% as well as the expected number of $\pw{T}$ gates required to execute the repeat-until-success program of $\sfrac{16}{5}$, 
fully automatically and without user-provided invariants.
The complete case study is presented in Appendix~\ref{app:repeat-until-success}.

\section{Related Work}
\label{s:related}

\textsf{Azure Quantum Resource Estimator}\cite{AQRE} and \textsf{Qualtran}~\cite{qualtran} provide hardware specific resource estimation for quantum programs.
These estimations are based on logical resources, such as the number of qubits and the number of \pw{T} gates, which are obtained from static programs (i.e.\@ in which control flow does not depend on measurements).
The tool \TOOL\ complements existing resource estimation tools by providing expected cost analysis (including logical resources) for programs with mid-circuit measurements and classical control flow.

% \textsf{QCEC}~\cite{qcec} from the \textsf{Munich Quantum Toolkit}~\cite{mqt} and \textsf{Feynman}~\cite{feynman} implement equivalence checking for circuits with no support for classical control flow.
% This is particularly relevant for compiler optimisations and circuit transformations. 
%
% \textsf{AutoQ}~\cite{ChenCLLT23} provides verification of static circuits based on tree automata techniques and leverages \textsf{SMT} solvers to check the generated verification conditions.
% Here the user provides pre- and postconditions. 
%
\ms{expanded this: includes silver and symqv}
The tools \textsf{SymQV}~\cite{BauerMarquartLS23} and \textsf{Silver}~\cite{LewisZS24} use symbolic execution to infer verification conditions and leverage \textsf{SMT} solvers.
\textsf{SymQV} provides verification of quantum circuits and supports measurements but no classical control flow.
\textsf{Silver} provides verification of \textsf{Silq}~\cite{silq20} programs and supports mid-circuit measurements and classical conditions but not while loops.
Both tools focus on verification of functional properties.
This is performed by encoding symbolic program states, preconditions, and the negation of the postcondition as \textsf{SMT} constraints and proving correctness by showing that the resulting formula is unsatisfiable.
The tool \TOOL\ uses symbolic execution to infer input-output relations of expected cost between program locations.
However, the tool focuses on expected cost and automatically infers a suitable \emph{witness}---a bounding function that serves as an upper invariant.
This synthesis problem is inherently more challenging and consequently the generated constraints are more complex and require advanced techniques for constraint solving.
Adding support for while loops further increases the complexity of the analysis as it requires suitable abstractions of loops. 
% A key difference is in the nature of the properties being verified or inferred and the support of loops. 
% In related tools, verification conditions are expressed over the symbolic program state and checked for validity.
% In contrast, \TOOL\ automatically infers a suitable \emph{witness}---a bounding function that serves as an upper invariant for the expected cost.
% This synthesis problem is inherently more challenging and consequently the generated constraints are more complex and require advanced techniques for constraint solving.
%
The state representation in related tools as well as \TOOL\ restricts to programs with a finite number of qubits and is fundamentally exponential, while various optimisations are applied to improve scalability.
In practice, \textsf{SymQV} and \textsf{Silver} can handle larger programs which is due to the increased complexity of expected cost analysis and support for while loops in \TOOL.
\textsf{AutoQ 2.0}~\cite{ChenCHHLLT25} (successor of \textsf{AutoQ}~\cite{ChenCLLT23}) provides automated verification of quantum programs using tree automata.
Here, the user provides pre- and postconditions for the quantum program, as well as loop invariants, which are then checked by the tool.
The tool leverages tree automata for efficient representation and manipulation of symbolic quantum states.
While the analysis is still limited to finite-dimensional state spaces, this approach scales in many examples better than the explicit symbolic representation employed in \TOOL.
\textsf{QPMC}~\cite{FengHTZ15} is a model checker for quantum programs and quantum protocols described as discrete-time Markov chains in which transitions are labelled with superoperators.
This tool focuses on the verification of quantum protocols, but also motivates verification of programs with loops.
Both tools provide automated verification of quantum programs, with mid-circuit measurements and loops.
While \TOOL\ targets a similar class of programs, it focuses on quantitative analysis, automatically inferring upper bounds on the expected cost.

% GM: reformulated
The formal framework introduced by Avanzini et al.\ in~\cite{AMPPZ:LICS:22,AMPP24}---dubbed \emph{quantum expectation transformer}---
forms an extension of Dijkstra's \emph{predicate transformer semantics}~\cite{Dijkstra75} and is strongly related to
Floyd-Hoare's logic~\cite{Hoare69} for reasoning about imperative programs.
These methodologies have been generalised to the \emph{weakest pre-expectation semantics}~\cite{GKM14} for probabilistic programs.
Kaminski et al.~\cite{KKMO:ACM:18,KKM20} investigate the expected runtime of probabilistic programs and introduces the \emph{expected cost transformer} for probabilistic programs.
% Automation of expectation transformers for probabilistic programs have been studied and realised by \cite{NCH18,AMS20,AMS23}.
%
Liu et al.~\cite{LiuZBY25} employ \emph{quantum weakest preconditions} to formally reason about expected runtime;
thereby extending earlier work~\cite{DHondtP06}.
However, automation is only considered superficially.
Our automation in \TOOL\ employs earlier automation efforts of \cite{KKMO:ACM:18}'s runtime transformer,
namely~\cite{NCH18,AMS20,AMS23}.

% \begin{figure}[t]
% % \hrulefill
% \[
%   \begin{array}{rlll}
%     \NExp &  \ni \nexp, \nexp_1, \nexp_2
%     & \rgl &   \x^{\NVar} \mid i \in \mathbb{Z} \mid
%              \nexp_1 + \nexp_2   \mid  \nexp_1 - \nexp_2  \mid \nexp_1 \times \nexp_2 \\
%     \BExp &  \ni \bexp, \bexp_1, \bexp_2
%     & \rgl &  \x^{\Bool}   \mid \true \mid \false \mid
%              \nexp_1 = \nexp_2  \mid  \nexp_1 < \nexp_2 \mid \neg \bexp \mid \bexp_1 \wedge \bexp_2 \mid \bexp_1 \vee \bexp_2 \\
%     \Exp &  \ni \e, \e_1, \e_2
%     & \rgl &  \nexp \mid \bexp \\
%     \Cmds &  \ni  \cmd, \cmd_1, \cmd_2
%     & \rgl &  \SKIP \mid  \x^{\K} <- \e^{\K}  \mid  \cmd_1; \cmd_2\mid\  \!  \IF \bexp^{\Bool} \THEN \cmd_1  \ELSE \cmd_2   \\
%           & & &\mid\ \WHILE \bexp^{\Bool} \DO \cmd
%                 \mid  \qs^{\Qubits}  <* \ope \mid   \x^{\Bool} <- \MEAS{\q^{\Qubits}} \mid \CONSUME~\nexp
%   \end{array}
% \]
% \vspace{-2\belowdisplayskip}
%  % \hrulefill
%  \caption{Syntax of the mixed classical-quantum imperative programming language \IMQ.}
% \label{fig:synt}
% \end{figure}

\section{Background}
\label{s:preliminaries}

% This work is concerned with automating the quantum expectation transformer of Avanzini et al.~\cite{AMPPZ:LICS:22,AMPP24}.
In this section, we recall key definitions and results from the literature.
We briefly describe the syntax of the mixed classical-quantum programming language \IMQ\ and present the quantum expectation transformer $\wpt{\cdot}{\cdot}$.
Due to space restrictions, the reader is kindly referred to~\cite{AMPPZ:LICS:22} for further details and motivations, as well
as to~\cite{MikeAndIke} for further background on quantum computing.
%
% Next, we discuss syntax and semantics of quantum programs.

\subsection{Mixed Classical-Quantum Programming Language}
\label{ss:lang}

\paragraph{Quantum Computing.}
Given a totally ordered set of qubits $Q=\{\q_1,\ldots,\q_n\}$, let $\mathcal{H}_{Q}$ be the $2^{n}$-dimensional Hilbert space defined by $\mathcal{H}_{Q} \triangleq \otimes_{i=1}^n \mathcal{H}_{\q_i}$, with $\mathcal{H}_{\q} = \Comp^{2}$ being the vector space with computational basis $\{\ket{0},\ket{1}\}$ and $\otimes$ being the tensor product. 
With $\bra{k}$ we denote the conjugate transpose of $\ket{k}$ for $k \in \{0,1\}$.
Let $\matrixspace$ be the set of complex square matrices acting on the Hilbert space $\mathcal{H}_{Q}$, i.e., $\matrixspace=\Comp^{2^n \times 2^n}$. 
Given $M \in \matrixspace$, $M^\dagger$ denotes the conjugate transpose of $M$, and $I_{2^n}$ denotes the identity matrix over $\matrixspace$. 
We will write $I$ when the dimension is clear from the context. 
%
%Following~\cite{AMPP24}, the \emph{quantum state} is modelled as a density operator.
Let $\densop \subsetneq \matrixspace$ be the set of all \emph{density operators}, i.e., positive semi-definite matrices of trace equal to $1$ on $\mathcal{H}_{Q}$. 
% Density operators can be viewed as the mathematical representation of a (mixed) quantum state. 
A \emph{unitary operator} $U$ is a matrix in $\matrixspace$ such that $UU^\dagger = U^\dagger U = I$.
% A \emph{superoperator} $\Phi_U : \densop \to \densop $, an endomorphism over density operators, is attached to each unitary operator $U$  and defined by $\Phi_U \triangleq \lambda \rho.U\rho U^\dagger$.
We denote the application of a unitary by $\Phi_U \triangleq \lambda \rho.U\rho U^\dagger$.
% By definition, $\Phi_U$ is a completely positive, trace-preserving linear map. Indeed, $tr(U\rho U^\dagger)=tr(\rho)$, by unitarity. 
% Hence, $U\rho U^\dagger$ is a density operator in $\densop$ for each $\rho \in \densop$.
%
Regarding measurements, for each $i$, $1 \leqslant i \leqslant card(Q)$, we define $\meas{k}{i} \in \matrixspace$, with $k \in \{0,1\}$, by $\meas{0}{i} \triangleq  I_{2^{i-1}} \otimes (\ket{0}\bra{0}) \otimes I_{2^{n-i}}$ and $\meas{1}{i} \triangleq I - \meas{0}{i}$. 
The measurement of the qubit $q_i$ (in the computational basis) of a density matrix $\rho \in \densop$ produces the classical outcome $k\in \{0,1\}$ with probability $tr(\meas{k}{i}\rho)$. 
The measurement $m_{k,i}$, is defined by
$\frac{\meas{k}{i}\rho\meas{k}{i}^\dagger}{tr(\meas{k}{i}\rho)}$ if $tr(\meas{k}{i}\rho) \neq 0$
and $\frac{I}{2^n}$ otherwise.

\paragraph{Syntax.}
The \IMQ\ programming language is an imperative mixed classical-quantum language that extends Dijkstra's Guarded Command Language~\cite{Dijkstra75} with quantum features.
% We briefly summarise its syntax here. 
The classical component is largely standard, while the quantum aspects require more detailed attention. 
The concrete syntax is available in Appendix~\ref{app:syntax}. Partly, the \IMQ's syntax was already highlighted in the motivation example above, cf.~Figure~\ref{fig:X2}.
Usually we indicate the type of a variable: Boolean $\Bool$, integer $\NVar$, qubit $\Qubits$.
% Further, let $\K$ denote an arbitrary classical type in $\{\Bool,\NVar\}$.
% For a comprehensive description, we refer the reader to~\cite{AMPP24}. 
% The language is typed and supports standard classical operations, including common Boolean and integer expressions.
Program statements are either classical assignments, conditionals, sequences, loops, 
\emph{quantum assignments} 
$\qs^{\Qubits} <* \ope$, 
or \emph{measurements} 
$\x^{\Bool} <- \MEAS{\q^{\Qubits}}$.
% A quantum assignment consists in the application of a quantum unitary gate $\ope$ of arity $ar(\ope)$ to a sequence of qubits $\qs \triangleq \q_1,\ldots,\q_{ar(\ope)}$.
% A unitary matrix $U$ will be associated with each quantum gate $\ope$.
% A measurement performs a single-qubit measurement of $\q$ in the computational basis: the outcome is a Boolean value and the quantum state evolves accordingly. 
% For a given syntactic construct $t$, let $\Bool(t)$ (respectively $\NVar(t)$, $\Qubits(t)$) be the set of Boolean (respectively, integer, qubit) variables in $t$.
Furthermore, code is instrumented by $\pw{\CONSUME~\nexp}$ statements to model resource cost.

\paragraph{Semantics.}

\newcommand{\cst}{s}
\newcommand{\cstt}{r}
\newcommand{\qst}{\rho}
\newcommand{\tcon}[1]{#1^\dagger}
\newcommand{\cfg}[3]{(#1,#2,#3)}
\newcommand{\cfgc}[2]{(#1,#2)}
\newcommand{\halt}{\downarrow}
\newcommand{\dist}[1]{\{#1\}}
\newcommand{\dirac}[1]{\{ 1: #1 \}}

A \emph{program state} in \IMQ\ is represented by a pair $\sigma \triangleq (\cst, \rho) \in \AState$ consisting of a classical state and a quantum state.
 
The \emph{classical state} is modelled as a (well-typed) \emph{store} $s$, that is, a mapping from variables to values.
Let $\cst[\x := k]$ be the store obtained from $s$ by updating the value assigned to $\x$ in the map $s$.
Given a store $\cst$, let $\sem[s]{-}$ denote the interpretation, mapping expressions to values, defined in the obvious way.
For example $\sem[\cst]{\x} \triangleq \cst(\x)$, $\sem[\cst]{i} \triangleq i$ for $i \in \mathbb{Z}$, $\sem[\cst]{\nexp_1 + \nexp_2} \triangleq {\sem[\cst]{\nexp_1} + \sem[\cst]{\nexp_2}}$, etc.
The classical assignment $\x := \e$ updates the store $\cst$ by assigning the value of expression $\e$ to variable $\x$, denoted $\cst[\x := \sem[\cst]{\e}]$.
 
Following~\cite{AMPP24}, the \emph{quantum state} is modelled as a density operator.
The quantum assignment $\qs <* \oper{U}$ updates the quantum state $\qst$ to a new quantum state $\Phi_{U_{\qs}}(\qst) = U_{\qs} \rho U_{\qs}^\dagger$, where $U_{\qs}$ is the unitary operator  in $\matrixspace$ computed by extending the quantum gate $\oper{U}$ to the entire set of qubits $Q$.
The measurement $\x <- \MEAS{\q_i}$ updates the state depending on the two possible outcomes, $k=0$ and $k=1$: 
with probability $tr(\meas{k}{i} \qst)$, the classical state is updated by $\cst[\x := k]$ and 
the quantum state is set to $m_{k,i}(\qst)$, where $m_{k,i}(\qst)$ is the (normalised) quantum state after the measurement.

We introduce the notion of expected cost and expected value informally, for a more formal treatment we refer to~\cite{AMPPZ:LICS:22}.
The \emph{operational semantics} can be formally expressed as a weighted relation over distributions of configurations.
Let $(\cmd, \sigma)$ denote the initial configuration.
The \emph{expected cost} $\ecost_{\cmd}(\sigma)$ is the expected total cost accumulated over all possible execution paths of $\cmd$ starting from $\sigma$, weighted by their respective probabilities.
The \emph{expected value} $\evalue_{\cmd}(f)(\sigma)$ is the expected value of a given expectation function $f$ over all possible terminal states reached by executing $\cmd$ from $\sigma$, weighted by their probabilities.

% Let $\lmulti p_i : \cfg{\cmd_i}{\cst_i}{\qst_i} \rmulti \in \delta$ be a configuration in $\delta$ with probability $p_i$.
% The weight for executing this configuration is $p_i \cdot \max(0,\sem[\cst_i]{\nexp})$ if $stmt_i = \CONSUME~\nexp$, and $0$ otherwise.
% % Here $p_i$ is the probability of executing the statement---specifically, it is the probability of the configuration within the distribution that is reduced.
% The terminal configuration is denoted by $\downarrow$.
%
% \newcommand{\Alg}{\overline{\Q}}
% \newcommand{\FADists}{\Dists_{\RAlg^+}^{\text{fin}}}
%
% For automation, we consider the expected cost and the expected value with respect to an expectation function $f$.
% The \emph{expected cost} is defined as the sum of all weights collected during execution starting from an initial configuration $(\cmd, \sigma)$:
% \[
%   \textstyle \ecost_{\cmd}(\sigma) \triangleq \sup_{n \in \mathbb{N}}\{\sum_{i=0}^{n} w_i \mid \lmulti 1: (\cmd, \sigma) \rmulti \toomqw{w_0} \cdots \toomqw{w_n} \delta \}  \tpkt
% \]
% The \emph{expected value} is defined as the expectation $f$ over terminal states:
% \[
%   \textstyle \evalue_{\cmd}(f)(\sigma)\triangleq \sup_{n \in \mathbb{N}} 
%   \{\sum \{ p_i \cdot f(\sigma_i) \mid \lmulti 1: (\cmd, \sigma) \rmulti \toomqw{w_0} \cdots \toomqw{w_n} \delta, \lmulti p_i : (\downarrow, \sigma_i) \rmulti \in \delta \}  \} \tpkt
% \]

\subsection{Quantum Expectation Transformers}
\label{ss:wp}

\ms{R1: technical dense}

\begin{figure*}[t]
% \hrulefill
\begin{align*}
  \wpt{\SKIP}{f} & \triangleq f
  \\
  \wpt{\CONSUME~\nexp}{f} & \triangleq \max(\lambda\_ .\ 0,\sem{\nexp}) \mathrel{+} f
  \\
  \wpt{\x <- \e}{f} & \triangleq f[\x := \e]
  \\
  \wpt{\cmd_1 ; \cmd_2}{f} &\triangleq \wpt{\cmd_1}{\wpt{\cmd_2}{f}}
  \\
  \wpt{\IF \bexp \THEN \cmd_1 \ELSE \cmd_2}{f} & \triangleq \wpt{\cmd_1}{f}\up{\sem{\bexp}} \wpt{\cmd_2}{f}
  \\
  \wpt{\WHILE \bexp \DO \cmd}{f} &\triangleq \lfp\left(\lambda F. \wpt{\cmd}{F} \up{\sem{\bexp}} f \right)
  \\
  \wpt{\qs <* \ope }{f} &\triangleq f[\Phi_{U_{\qs}}]
  \\
  \wpt{\x <- \MEAS{\q_i}}{f}
  & \triangleq f[\x := 0\text{; }m_{0,i}] \up{\proba{0}{i}} f[\x := 1\text{; }m_{1,i}]
\end{align*}
\vspace{-2\baselineskip}
% \hrulefill
% \vspace{-2mm}
\caption{The quantum expectation transformer $\wpt{\cdot}{\cdot}$.}
\label{fig:qwpt}
\vspace{-1\baselineskip}
\end{figure*}

% We recall the notion of quantum expectation transformers, cf.~\cite{AMPPZ:LICS:22,AMPP24}.
% This forms the formal framework for automating the expected cost analysis.

% \paragraph{Expectations and Expectation Transformer.}

The \emph{quantum expectation transformer} presented by Avanzini et al.~\cite{AMPPZ:LICS:22,AMPP24} is a formal framework for reasoning about expectations of quantum programs.
This transformer is a mapping from \emph{expectations}---real-valued functions on program states---to expectations in a continuation-passing style, formally:
$\wpt{\cdot}{f} : \Cmds \to (\AState \to \Rext) \to (\AState \to \Rext)$.
% \[\wpt{\cdot}{f} : \Cmds \to (\AState \to \Rext) \to (\AState \to \Rext)\tpkt\]
The transformer is defined inductively in Figure~\ref{fig:qwpt}, and serves as a semantic model to express properties on expectations, including both expected cost and expected value.
% An \emph{expectation} $f$ is a real-valued function from states $\AState$ to the extended non-negative reals $\Rext \triangleq [0,+\infty]$.
% The quantum expectation transformer is a mapping from \emph{expectations}---real-valued functions on program states---to
% expectations in a continuation-passing style,
% \[
%   \wpt{\cdot}{f} : \Cmds \to (\AState \to \Rext) \to (\AState \to \Rext)
%   \tkom
% \]
% An \emph{expectation} is a function that assigns a real value to each program state, representing quantities of interest such as cost, probability, or other measurements. In the context of quantum programs, expectations can depend on both classical and quantum components of the state. The \emph{expectation transformer} $\wpt{\cdot}{\cdot}$ is a semantic function that, given a program statement and a post-expectation, computes the pre-expectation: it describes how the execution of the statement transforms expectations backward through the program. This approach generalises classical predicate transformers to the quantum setting, allowing reasoning about expected values and costs in programs that combine classical and quantum effects.
% The transformer is defined inductively on statements in Figure~\ref{fig:qwpt}.
 
For any expression $\e$, $\sem{\e}$ is a shorthand notation for the function $\lambda \ltuple \cst, \qst \rtuple. \sem[\cst]{\e} \in \AState \to \Rext$.
We also use $f[\x := \e]$ for the expectation $\lambda \ltuple \cst, \qst \rtuple.f\ltuple \cst[\x := \sem[\cst]{\e}], \qst \rtuple$.
For a given map $\qmap : \densop \to  \densop$,  we define $f[\qmap] \triangleq \lambda (\cst,\qst). f(\cst, \qmap(\qst) )$. 
We group such state modifications, for instance, $f[\x := \e\text{; }\qmap]$ stands for $(f[\x := \e])[\qmap]$
and $f[\x := \e,\y:=\e']$ stands for $(f[\x := \e])[\y:=\e']$.
For $p  \in \AState \to [0,1]$ and $f,g \in \AState \to \Rext$, $f \up{p} g$ denotes the
function $\lambda \sigma. p(\sigma) \cdot f(\sigma) + (1-p(\sigma)) \cdot g(\sigma) \in \AState \to \Rext$,
similar we use $f \cdot g$ to denote $\lambda \sigma. f(\sigma) \cdot g(\sigma) \in \AState \to \Rext$.
Thus, for instance, $f[\x := \x + 1] +_{\sem{\x = 1}} f$ behaves like~$f$, except that $\x$ is first incremented when applied to states
with classical variable $\x$ equal to $1$.
In correspondence to the normalisation of quantum state $m_{k,i}$, we define probabilities $\proba{k}{i} \triangleq \lambda \qst. tr(\meas{k}{i}\qst\meas{k}{i}^{\dagger})$.
We overload this function from $\densop$ to $\AState$ so that $\proba{k}{i}(\cst,\qst) = \proba{k}{i}(\qst)$.
In this way, $f[\x := 0\text{; }m_{0,i}] \up{\proba{0}{i}} f[\x := 1\text{; }m_{1,i}]$
computes the expected value of $f$ on the distribution of states obtained by measuring the $i$-th qubit and
assigning the outcome to classical variable $\x$.
In the case of loops, the least fixed point $\lfp$ is defined with respect to the pointwise ordering on the function space $\AState \to \Rext$.

\subsection{Towards Automation}

\begin{figure*}[t]
% \hrulefill
\begin{alignat*}{2}
 \quad
& \!\!\!\!\!\!\law[idents:mono]{monotonicity}   && g \geq f \Rightarrow \wpt{\cmd}{g} \geq \wpt{\cmd}{f} \\
& \!\!\!\!\!\!\law[idents:ui]{upper invariance} &~& (g \geq \sem{\bexp}{\cdot}\wpt{\cmd}{g} \wedge g \geq \sem{\neg \bexp}{\cdot} f )\! \Rightarrow \! g \geq \wpt{\WHILE \bexp \DO \cmd}{f} \\
& \!\!\!\!\!\!\law[idents:rext-sep]{separation} && \qect{\cmd}{f} = \qect{\cmd}{\lambda \sigma .\ 0} \up{} \qevt{\cmd}{f}%
\end{alignat*}%
% \hrulefill
\vspace{-2\baselineskip}
\caption{Universal laws derivable for the quantum expectation transformer.}
\label{fig:idents}
\end{figure*}

Next, we provide key concepts for automating the expected cost analysis of quantum programs.
Let $\mathsf{CostFree}(\cmd)$ replace $\pw{\CONSUME~\nexp}$ with $\pw{\SKIP}$ in $\cmd$.
The \emph{quantum expectation cost transformer} and the \emph{quantum expectation value transformer} are:
% \begin{align*}
%   \qect{\cmd}{f} & \triangleq \wpt{\cmd}{f} \\
%   \qevt{\cmd}{f} & \triangleq \wpt{\mathsf{CostFree(\cmd)}}{f} \tpkt
% \end{align*}
% GM: space, but looks strange
\begin{equation*}
  \qect{\cmd}{f} \triangleq \wpt{\cmd}{f} \quad
  \qevt{\cmd}{f} \triangleq \wpt{\mathsf{CostFree(\cmd)}}{f} \tpkt
\end{equation*}%
\begin{proposition}
  \label{t:adequacy}
For all statements $\cmd \in \Cmds$ and expectations $f : (\AState \to \Rext)$:
% GM: space  
\begin{equation*}
  \qect{\cmd}{\lambda \_ . 0}  = \ecost_{\cmd} \quad
  \qevt{\cmd}{f}  = \evalue_{\cmd}(f)
  \tpkt
\end{equation*}
\end{proposition}
Proposition~\ref{t:adequacy} has been stated and proved originally in~\cite{AMPPZ:LICS:22}.
Crucially, it is enough to focus on the quantum expectation transformer to reason about both the expected cost and the expected value of quantum programs.
Figure~\ref{fig:idents} depicts a couple of useful laws, which have been stated and proved in~\cite{AMPPZ:LICS:22}.
We denote by $\geq$ also the pointwise extension of the order from $\Rext$ to functions, that is,
$f \geq g$ holds iff $\forall \sigma \in \AState,\ f(\sigma) \geq g(\sigma)$.
The \ref{idents:mono} Law permits us to reason modulo upper bounds---actual expectations can be always substituted by upper bounds.
The \ref{idents:ui} Law constitutes a generalisation of the notion of invariant stemming from Hoare calculus.
It is used to find closed-form upper bounds $g$ to expectations $f$ of loops and is central to our approach for automating the cost analysis of quantum programs.
The \ref{idents:rext-sep} Law allows us to decompose the expected cost analysis into an expected cost and an expected value analysis.
This modularity is crucial for automation, especially when dealing with complex loop structures and larger quantum programs. 
We emphasise that the composability of the semantic expectation transformer (see Figure~\ref{fig:qwpt}) is essential here.

% We comment on the choice of quantum state representation:
% Quantum states can be represented in different ways, most notably as state vectors (pure states) or as density operators (pure and mixed states, i.e.\xspace a probabilistic mixture of pure states). 
% The transformer framework is valid for both representations.
% In this work, we apply the density operator formalism.
% \cite{qpl} motivates the density operator representation in his seminal work on quantum programming languages: 
% Density operators
% \begin{inparaenum}[(i)]
% \item frequently allow for simpler mathematical expressions,
% \item have a unique representation of pure states, and
% \item provide an efficient representation of mixed states.
% \end{inparaenum}
%
% For the automation, we perform symbolic evaluation to derive input-output relations.
% Our initial observations confirmed remark (i), and ambiguity in (ii) may introduce unnecessary complexity when working with symbolic expressions.
% Therefore, we have chosen to use density operators.
% Notice that the program semantics is defined via distributions over density operators, rather than using a mixed state representation (cf.\xspace \cite{qpl,Perdrix08b}).
% In the automation, however, we consciously refrain from enforcing the purity of density operators.
% % \GM{correct?}
% % \MS{yes}
% For automation, we need to strike a balance between expressiveness and complexity.
% As we focus on upper bounds, it is safe to over-approximate reachable states and relax the constraints on the state representation.
%

In this work, we apply and abstract the density operator formalism.
Selinger \cite{qpl} motivates the density operator representation in his seminal work on quantum programming languages: 
% Density operators
% \begin{inparaenum}[(i)]
(i) frequently allows for simpler mathematical expressions,
(ii) has a unique representation of pure states, and
(iii) provides an efficient representation of mixed states.
% \end{inparaenum}
%
For automation, we introduce a syntactic expectation transformer (see Figure~\ref{fig:symbolic-inference}).
Our initial observations confirmed remark (i), and ambiguity in (ii) may introduce unnecessary complexity when working with symbolic expressions.
Therefore, we have chosen to use density operators.
% Observe, that the symbolic state is defined as partial mapping instead of a matrix.
% This provides a simple abstraction heuristic driven optimisation by restricting the state to the variables of interest.
However, in line with~\cite{AMPPZ:LICS:22,AMPP24}, semantics are defined via distributions over density operators rather than mixed state representations.

\section{Automation}
\label{s:automation}

\newcommand{\astin}{\symbolic{\astate_{in}}} % inital symbolic state
\newcommand{\stin}{\symbolic{\rho_{in}}} % inital symbolic quantum state

In this section, we present the main contribution of this work and discuss the implementation of \TOOL, an automated expected cost analysis tool for mixed classical-quantum programs written in \IMQ.
%
% This approach is inspired by related work on automated cost analysis of probabilistic programs~\cite{AMS20,AMS23} 
% and extends the quantum expectation transformer of \cite{AMPPZ:LICS:22,AMPP24} to a syntactic variant that is suitable for automation.
%
Refer back to Figure~\ref{fig:workflow} for a schematic overview.
For automation, we introduce a syntactic variant 
$\qinfer{\cmd}{f}$ (see Figure~\ref{fig:symbolic-inference}) of the expectation transformer 
$\wpt{\cmd}{f}$ (see Figure~\ref{fig:qwpt}) to infer constraints on symbolic expressions.
Crucially, its interpretation forms an upper bound $\sem{\qinfer{\cmd}{f}} \geq \wpt{\cmd}{\sem{f}}$.
%
% We fix a syntactical category for bounding functions, so-called \emph{cost expressions}, and more generally for \emph{expectations}, that is, real-valued functions on program states.
For loop-free code, $\qinfer{\cmd}{f}$ performs symbolic evaluation on a symbolic state representation to form input-output relations.
For while statements, we apply law~\ref{idents:ui} to infer term constraints:
    \begin{equation*}
\bexp \vdash G \geq \qinfer{\WHILE \bexp \DO \cmd}{G} \qquad
\neg \bexp \vdash G \geq F \tpkt
    \end{equation*}
%     \begin{align*}
% \bexp & \vdash G \geq \qinfer{\WHILE \bexp \DO \cmd}{G} \\
% \neg \bexp & \vdash G \geq F
%     \end{align*}
Here $G$ and $F$ represent unknown expectations.
To resolve the inequalities resulting from while statements, we provide a set of inference rules to reduce inequality constraints over expectations to inequality constraints over cost expressions (see Subsection~\ref{ss:inference-cost-constraints}).
These constraints are then further reduced to inequality constraints over polynomial expressions with unknown coefficients (see Subsection~\ref{ss:inference-polynomial-constraints}).
Finally, we use known approaches to synthesise certificates for the polynomial constraints, resulting in upper invariants for while statements (see Subsection~\ref{ss:inference-certificate-constraints}).
Given that, a bottom-up strategy can be implemented to infer upper invariants for loops and propagate the result (see Subsection~\ref{ss:algorithm}).

\subsection{Symbolic Representation}

In this subsection, we introduce the symbolic representation that underpins our approach to the automated inference. 

% GM: paragraphs
% \subsubsection{Symbolic Representation of Quantum States.}
\paragraph{Quantum states.}
% In preparation of the symbolic inference of upper invariant constraints, we introduce a symbolic representation of the quantum state.
% We aim for a uniform representation of the classical and quantum part of the program state.
% Informally, consider representing a density matrix as a mapping from variables to individual elements of the matrix and representing matrix operations as parallel updates.

Recall that an expectation is a function from program states $(\cst, \rho)$ to non-negative real numbers.
To ease the representation of expectations within the automation, we model complex numbers as pairs $(a,b)$ of real numbers, where $a$ is the real part and $b$ is the imaginary part.
%
% Consider a quantum state $\rho$ represented as a density matrix:
% \[
% \rho = \begin{pmatrix}
% (a_{11},b_{11}) & (a_{12}, b_{12}) & \cdots & (a_{1n},b_{1n}) \\
% (a_{21},b_{21}) & (a_{22}, b_{22}) & \cdots & (a_{2n}, b_{2n}) \\
% \vdots & \vdots & \ddots & \vdots \\
% (a_{n1},b_{n1}) & (a_{n2}, b_{n2}) & \cdots & (a_{nn},b_{nn})
% \end{pmatrix}
% \]
Since $\rho$ is a density matrix, the matrix is Hermitian, that is, $\rho = \rho^\dagger$.
The diagonal elements are constrained by $a_{ii} \in \R$, $a_{ii} \geqslant 0$,  $\sum_{i=1}^n a_{ii} = 1$ and $b_{ii} = 0$.
The off-diagonal elements are constrained by $a_{ij} \in \R$ and $b_{ij} \in \R$, $a_{ij} = a_{ji}$ and $b_{ij} = -b_{ji}$, for $i \neq j$.
Due to the symmetry we restrict off-diagonal elements to $a_{ij}$ and $b_{ij}$ with $i < j$.
Given a program acting on the Hilbert space $\CPLX^{m} \times \CPLX^{m}$ for $m = 2^n$, we fix the set of \emph{density matrix variables} by 
$d_i$ for diagonal elements and $a_{ij}, b_{ij}$ for $1 \leqslant i < j \leqslant m$ for off-diagonal elements.
In what follows, let $\stin$ denote the symbolic density matrix:
\ms{rename to $\stin$ since index is later used for branching @MS check consistency}
\[
\stin \triangleq \begin{pmatrix}
(d_{1},0) & (a_{12}, b_{12}) & \cdots & (a_{1m},b_{1m}) \\
(a_{12},-b_{12}) & (d_{2},0) & \cdots & (a_{2m}, b_{2m}) \\
\vdots & \vdots & \ddots & \vdots \\
(a_{1m},-b_{1m}) & (a_{2m}, -b_{2m}) & \cdots & (d_{m},0)
\end{pmatrix}
\]

% Notice that the program semantics is defined via distributions over density operators, rather than using a mixed state representation (cf.\xspace \cite{qpl,Perdrix08b}).
% In the automation, however, we consciously refrain from enforcing the purity of density operators.
% % \GM{correct?}
% % \MS{yes}
% For automation, we need to strike a balance between expressiveness and complexity.
% As we focus on upper bounds, it is safe to over-approximate reachable states and relax the constraints on the state representation.

% \subsubsection{Symbolic Representation of Arithmetic Expressions.}
\paragraph{Arithmetic expressions.}

We restrict the set of applicable gates to the (approximative) universal set of \cliffordt\ gates, cf.~\cite{MikeAndIke}.
Further, we include gates that can be derived such as $\pw{CCNOT}$.
This restriction is not a limitation of the transformer framework, but rather a design choice to restrict the form of symbolic expressions.
The most complex gate in terms of representation is the $\pw{T}$ gate:
\[
    \pw{T} \triangleq 
    \begin{pmatrix}
        1 & 0 \\
        0 & \frac{1+i}{\sqrt{2}}
    \end{pmatrix}
    \triangleq 
    \begin{pmatrix}
        (1,0) & (0,0) \\
        (0,0) & (\frac{1}{\sqrt{2}}, \frac{1}{\sqrt{2}})
    \end{pmatrix}
\]
Taking matrix multiplication and state normalisation into account,
we represent arithmetic expressions, denoted as~$\AExp$, over the set of matrix and integer variables by:
% \[
%   \sfrac{e}{f}\text{\;, $e$ and $f$ are polynomial expressions with coefficients in $\Z[\sqrt{2}]$}
%   \tpkt
% \]
$
  \sfrac{e}{f}\text{\;, $e$ and $f$ are polynomial expressions with coefficients in $\Z[\sqrt{2}]$}
$.
% Clearly the denominator should never evaluate to zero for any assignment.
% We comment on the well-definedness in Subsection~\ref{ss:implementation}, when discussing details of the implementation.

% \subsubsection{Symbolic Representation Quantum Operations.}
\paragraph{Quantum operations.}

% Let $\QVARS$ denote the set of density matrix variables of a quantum program.
Next, we provide a symbolic representation of quantum operations as parallel updates.
% Informally, consider introducing variables $m_{ij}$ for each each matrix entry. 
% A matrix multiplication can then be represented as parallel assignment $m_{ij} = (M_1 M_2)_{ij}$.
Let $f$ denote an operator from density matrices to density matrices $\rho' = f(\rho)$. 
In particular, set $f(\rho) = U \rho U^\dagger$ for any unitary $U$ and 
$f(\rho) = m_{k,i}(\rho)$ for any post-measurement operator $m_{k,i}$.
We define the application of $f$ as an element-wise mapping from matrix variables to the upper triangle elements of $\symbolic{\rho}' = f(\symbolic{\rho})$, 
such that $d_i$ maps to $\mathsf{Real}(\symbolic{\rho_{ii}}')$, $a_{ij}$ maps to $\mathsf{Real}(\symbolic{\rho_{ij}}')$ and $b_{ij}$ maps to $\mathsf{Imag}(\symbolic{\rho_{ij}}')$. 
For the rest of the paper we denote the update mapping induced by $\symbolic{\rho}' = f(\symbolic{\rho})$ as $\{ \qs \mapsto f(\symbolic{\rho})\}$.

\begin{example} 
    % \begin{gather*}
    % \symbolic{\rho_0}=\begin{pmatrix}(d_1,0)&(a_{12},b_{12})\\(a_{12},-b_{12})&(d_2,0)\end{pmatrix},\qquad
    % H=\tfrac{1}{\sqrt{2}}\begin{pmatrix}(1,0)&(1,0)\\(1,0)&(-1,0)\end{pmatrix}\\
    % H\symbolic{\rho_0}H^\dagger=\tfrac{1}{2}\begin{pmatrix}(d_1+d_2+2a_{12},0)&(d_1-d_2,-2b_{12})\\(d_1-d_2,2b_{12})&(d_1+d_2-2a_{12},0)\end{pmatrix}
    % \end{gather*}
    % The corresponding parallel update is compactly written as
    % \[
    % d_1\mapsto\tfrac{1}{2}(d_1+d_2+2a_{12}),\quad
    % d_2\mapsto\tfrac{1}{2}(d_1+d_2-2a_{12}),\quad
    % a_{12}\mapsto\tfrac{1}{2}(d_1-d_2),\quad
    % b_{12}\mapsto -b_{12}.
    % \]
    Consider the application of the Hadamard gate $\pw{H}$ to a one-qubit system with symbolic density matrix $\stin$.
    \begin{gather*}
        \stin \triangleq \begin{pmatrix}
            (d_1,0) & (a_{12}, b_{12}) \\
            (a_{12}, -b_{12}) & (d_2,0)
        \end{pmatrix}  \qquad 
        \pw{H} \triangleq \frac{1}{\sqrt{2}}\begin{pmatrix}
            (1,0) & (1,0) \\
            (1,0) & (-1,0)
        \end{pmatrix} \\
        H \stin H^\dagger = 
        \frac{1}{2} \begin{pmatrix}
        (d_1 + d_2 + 2a_{12},0) & (d_1 - d_2, -2b_{12}) \\
        (d_1 - d_2, 2b_{12}) & (d_1 + d_2 - 2a_{12},0)
        \end{pmatrix} 
    \end{gather*}
    The inferred update mapping $\{\qs \mapsto H \stin H^\dagger\}$ of this matrix operation is:
    % \begin{align*}
    %     d_1 & \mapsto \frac{1}{2}(d_1 + d_2 + 2a_{12}) &
    %     d_2 & \mapsto \frac{1}{2}(d_1 + d_2 - 2a_{12})  \\
    %     a_{12} & \mapsto \frac{1}{2}(d_1 - d_2) &
    %     b_{12} & \mapsto \frac{1}{2}(-2b_{12})
    % \end{align*}
    \[
    d_1\mapsto\tfrac{1}{2}(d_1+d_2+2a_{12}) \quad
    d_2\mapsto\tfrac{1}{2}(d_1+d_2-2a_{12}) \quad
    a_{12}\mapsto\tfrac{1}{2}(d_1-d_2) \quad
    b_{12}\mapsto -b_{12}
    \]
\end{example}

% \subsection{Term Representation}

% This subsection provides a syntactical representation of expectations.
% We typically present substitutions as mapping $\{ x_1 \mapsto t_1, \ldots, x_n \mapsto t_n \}$.

% GM: subsubsection -> paragraph
\paragraph{Symbolic State.}
In the constraint inference, symbolic states are used to generate input-output relations between control flow locations.
Let $\SVARS$ denote the set of program variables and density matrix variables.
Adapting the notion of quantum program states of Section~\ref{s:preliminaries},
we define a symbolic state $\symbolic{\astate}$ as a partial store 
$\symbolic{\astate} \ofdom \SVARS \rightharpoonup \SExp$ 
from variables to expressions.
% This includes mapping Boolean variables $\BVARS$ to $\BExp$, numeric variables $\NVARS$ to $\NExp$ and density matrix variables $\QVARS$ to $\AExp$.
%
Let $\symbolic{\astate} \in \symbolic{\AState}$ denote a symbolic state.
We define the substitution $\symbolic{\astate}\{ \x \mapsto \e \}$ as the symbolic state obtained by substituting every occurrence of $\x$ in expressions of $\symbolic{\astate}$ with $\e$. 
Formally, for all variables $y \in dom(\symbolic{\astate})$,
$
\big(\symbolic{\astate}\{ \x \mapsto \e \} \big)(y) \triangleq \symbolic{\astate}(y)\{\x \mapsto \e \}
$.
%
% This mapping is partial, which is an implementation detail for automation that provides a simple abstraction mechanism by restricting the state to the variables of interest.
% For example, a symbolic store with domain $d_1, \ldots, d_{m}$ only infers the input output relation of the diagonal elements of the density matrix.
%
Let $\SVARS' \subseteq \SVARS$ be a subset of variables.
The initial symbolic state $\astin$ maps every variable $\x{} \in \SVARS'$ to itself.

\TOOL\ implements a heuristic-driven approach in which often only a subset of variables is tracked in symbolic states, such as, a small set of diagonal/off-diagonal entries.
This greatly reduces expression size and improves scalability.

\paragraph{Context.}
The inference system generates constraints of the form $\phi \vdash e \geq f$.
Here $\phi$ is a Boolean formula that reflects refinements due to control flow and assumptions in expectation expressions.
%
% We extend $\BExp$ to include inequalities over arithmetic expressions in $\AExp$ and statements about probabilities.
This includes inequalities over arithmetic expressions and statements about probabilities.
The latter is used for case distinction in post-measurement operations.
% More formally, $\phi$ is a Boolean expression over atoms:
% \begin{align*}
%     & \bvar \in \BVARS 
%     \quad \nexp_1, \nexp_2 \in \NExp
%     \quad \aexp, \aexp_1, \aexp_2 \in \AExp
%     \\
%     & 
%     \bot \mid
%     \top \mid
%     \bvar \mid
%     \nexp_1  \igeq \nexp_2 \mid
%     \aexp_1 \rgeq  \aexp_2 \mid
%     \aexp_1 \rgt   \aexp_2 \mid
%     \isProb{\aexp} \mid
%     \isProbOne{\aexp}
% \end{align*}
% We distinct between inequalities over classical program arithmetic expressions in $\BExp$, which arise for example from loop guards, and inequalities over arbitrary arithmetic expression in $\AExp$, which for example arise due to assumptions in cost expressions.
% We typically write
% $\nexp_1 \igt \nexp_1$ for $\nexp_1 \igeq \nexp_2 - 1$ and
% $\nexp_1 \ieq \nexp_1$ for $\nexp_1 \igeq \nexp_2 \wedge \nexp_2 \igeq \nexp_1$.
% Similarly, we write
% $\aexp_1 \req \aexp_2$ for $\aexp_1 \rgeq \aexp_2 \wedge \aexp_2 \rgeq \aexp_1$.
% Further, 
% $\neg (\aexp_1 \igeq \aexp_2) \triangleq \aexp_1 \igt \aexp_2$ and
% $\neg (\aexp_1 \rgeq \aexp_2) \triangleq \aexp_2 \rgt \aexp_1$.
%
% For the post-measurement operation the quantum state is normalised and elements of the density matrix are divided by its probability.
% For well-definedness we perform a case distinction on the probability expression.
The probability expression $\aexp \in \AExp$
% evaluates to zero $\isProbZero{\aexp}$, 
evaluates to one $\isProbOne{\aexp}$, or strictly between zero and one $\isProb{\aexp}$.
Furthermore, we set $\neg \isProbOne{\aexp} = \neg \isProb{\aexp} = \bot$.

To decide whether $\phi \vdash e \geq f$ holds we make use of additional assumptions that arise from the fact that a quantum state is a density matrix.
% GM repetition
% In particular all diagonal elements are non-negative real numbers that sum up to one,
% and all off-diagonal elements are real numbers within the range $[-1,1]$.
We typically keep these constraints implicit if not important.
Further constraints could be imposed.
For instance, $\rho = \rho^2$ restricts to pure quantum states.
However, we deliberately relax the constraints imposed on quantum states for automation.

\paragraph{Cost Expressions.}
% As an intermediate representation the inference system generates constraints of the form $\somecostconstraint$.
A cost expectation is a function from program states to non-negative real numbers.
For automation, we introduce a class of bounding functions, dubbed \emph{cost expressions} (cf.~\cite{pubs,AMS20,AMS23}).
Let $\mathsf{p}_i$ (probability) and $\mathsf{w}_i$ (weight) denote arithmetic expressions that evaluate into the non-negative reals.
We introduce cost expressions syntactically as linear combination of norms:
\[ \textstyle \Norm \ni \normexp \defsym \guarded{\phi} \cdot \aexp \qquad \CExp \ni \cexp \defsym \sum_i \mathsf{p}_i \cdot \mathsf{w}_i \cdot \normexp_i \]
% A norm is an arithmetic expression that evaluates to a non-negative number and can be conceived as an abstraction of the program state.
A norm is a guarded arithmetic expression $\guarded{\phi} \cdot \aexp$ that evaluates to $\aexp$ if $\phi$ holds, and $0$ otherwise.
We require that $\aexp$ is always non-negative in the given context.
In \TOOL, norms are chosen via heuristics and have the form $\guarded{\aexp \rgeq 0} \cdot \aexp \triangleq \mathsf{max}(\aexp,0)$.
Cost expressions are closed under addition and scalar multiplication with expressions that evaluate to a non-negative number (like probabilities).
We set 
$\guarded{\phi} \cdot (\sum_i \mathsf{p}_i \cdot \mathsf{w}_i \cdot ( \guarded{\phi_i} \cdot \aexp_i)) \triangleq \sum_i \mathsf{p}_i \cdot \mathsf{w}_i \cdot (\guarded{\phi \wedge \phi_i} \cdot \aexp_i)$.
The interpretation of cost expressions 
$\sem{\,\mathbin{\cdot}\,} \ofdom \CExp \to \AState \to \Rext$ 
is the obvious one.
%
% Next, we provide the syntactical representation of expectations.
\begin{figure*}[t]
  % \hrulefill
  \[
    \begin{array}{r@{\ }c@{\ }l r}
      \qinfer{\SKIP}{\texp}                                     & = & \texp \\[1mm]
      \qinfer{\CONSUME~\nexp}{\texp}                            & = & \TCONSUME~\nexp~\texp\\[1mm]
      \qinfer{\x{} <- \e}{\texp}                                  & = & \texp\{\x{} \mapsto \e \}\\[1mm]
      \qinfer{\cmd_1 ; \cmd_2}{\texp}                           & = & \qinfer{\cmd_1}{\qinfer{\cmd_2}{\texp}}\\[1mm]
     \qinfer{\begin{array}{@{}l@{}}\IF \bexp\\ \THEN \cmd_1\\ \ELSE \cmd_2\end{array}}{\texp} & = & \TCOND~\bexp\phantom{x}\qinfer{\cmd_1}{\texp}\phantom{x}\qinfer{\cmd_2}{\texp}\\[1mm]
      \qinfer{\qs <* \ope }{\texp}           & = & \texp\{\qs \mapsto U\stin U^\dagger\}\\[1mm]
      \qinfer{\x <- \MEAS{\q_i}}{\texp} & = & 
        \TMEASURE~
         (p_{0,i}(\stin),\texp\{\x \mapsto 0; \qs \mapsto m_{0,i}(\stin)\}) \\
        && \phantom{xxxx}\; (p_{1,i}(\stin),\texp\{\x \mapsto 1; \qs \mapsto m_{1,i}(\stin)\}) \\[1mm]
     \qinfer{\WHILE \bexp \DO \cmd}{\texp}        & = & \funzero,  ~ \text{with side-conditions}\\
     && \phantom{xx}\CC{\begin{cases} \phantom{\neg}\bexp \vdash \funzero \geq \qinfer{\cmd}{\funzero} \\ \neg \bexp \vdash \funzero \geq \texp\end{cases}}\\[1mm]
    \end{array}
  \]
  % \vspace{-2mm}
  % \hrulefill
  % \vspace{-2mm}
  \vspace{-\belowdisplayskip}
\caption{Term representations of $\qinfer{\cdot}{\cdot}$ and their corresponding side-conditions.}
\label{fig:symbolic-inference}
\end{figure*}

\paragraph{Expectation Terms.}

The term structure and the term interpretation are chosen such that they reflect the operations of the expectation transformer.
Further, terms include unknown expectations and cost expressions.
\begin{align*}
    \TExp \ni \texp_1, \texp_2 & \defsym 
    \TFUN~\funsym~\symbolic{\astate}
    \plusT \TCOST~\cexp 
    \plusT \TCONSUME~\nexp~\texp
    \plusT \TCOND~\bexp~\texp_1~\texp_2 
    \plusT \TMEASURE~(\aexp_1, \texp_1)~(\aexp_2, \texp_2) 
\end{align*}
The term $\TFUN~\funsym~\symbolic{\astate}$ represents an unknown expectation for some while command with location $\ell$.
This is used to form the upper invariant constraints for while statements.
All other terms reflect different operations of the expectation transformer.
For example $\TMEASURE~(\aexp_1, \texp_1)~(\aexp_2, \texp_2)$ indicates a measurement operation with two outcomes. 
With probability $\aexp_1$ ($\aexp_2$) the classical outcome of the measurement is $0$ ($1$) and $\texp_1$ ($\texp_2$) is the resulting expectation taking the respective post-measurement state into account.
% The case in which the probability of $\aexp_1$ or $\aexp_2$ is zero is handled by cased distinction.
%
% In \cite{AMS23} constraints are directly formed over cost expressions. This makes the representation more succinct.
% Here the term structure is used as an intermediate representation, which we consider to be useful for automation as it simplifies necessary case distinctions in measurements and allows for some simplifications.
%
Let $\alpha \ofdom \Funsym \to \symbolic{\AState} \to \CExp$ be an assignment of cost expectation expressions to function symbols $\Funsym$.
The interpretation of expectation terms under $\alpha$ is defined in the obvious way.
For example 
$\sem[\alpha]{\TFUN~\funsym~\symbolic{\astate}} = \sem[]{ \alpha(\funsym)(\symbolic{\astate})}$,
$\sem[\alpha]{\TCOND~\bexp~\texp_1~\texp_2} = \sem[\alpha]{\texp_1}(\astate)$ if $\sem{\bexp}(\astate)$ and $\sem[\alpha]{\texp_2}(\astate)$ otherwise, etc.

In \cite{AMS23} constraints are directly formed over cost expressions. This makes the representation more succinct.
In \TOOL\ the term structure is used as an intermediate representation, which is useful for automation as it simplifies tracking necessary case distinctions in measurements and allows for useful simplifications.
For brevity, we omit these simplifications here and refer to Appendix~\ref{app:term-simp}.

\subsection{Inference of Expectation Term Constraints}
\label{ss:inference-term-constraints}

The syntactic variant of the expectation transformer is depicted in Figure~\ref{fig:symbolic-inference}.
For while statements an unknown expectation $\funzero$ is introduced.
The location symbol $\ell$ is a fresh symbol that represents the location of the loop statement.
The initial symbolic state $\astin$ is the identity mapping for some subset of variables $\SVARS' \subseteq \SVARS$ and represents program states before evaluating the loop condition.
Further, side-conditions are introduced to form the constraints for the upper invariant.
For measurements both outcomes are considered.
The term $p_{k,i}(\stin)$ for $k \in \{0,1\}$ is the probability expression obtained by $p_{k,i}$ applied to the symbolic state $\stin$.
% GM rewritten
% While $\texp\{\x \mapsto k; \qs \mapsto m_{k,i}(\symbolic{\rho_0})\}$ performs the update to represent the state after the measurement.
On the other hand, $\texp\{\x \mapsto k; \qs \mapsto m_{k,i}(\stin)\}$ performs the update to represent the state after the measurement.
Similarly, $\texp\{\qs \mapsto U\stin U^\dagger\}$ performs the update to take the unitary operation into account.
We arrive at our main theoretical contribution, the soundness of the
inference algorithm with the proof in the Appendix~\ref{p:soundness}. In conjunction with Proposition~\ref{t:adequacy} this implies the soundness of
the expected cost analysis provided.
%
% \begin{theorem}[Soundness Theorem]
%   \label{t:soundness}
%   If all side-conditions in Figure~\ref{fig:symbolic-inference} are met, then for all programs
%   $\cmd$, $\lambda f. \qet{\cmd}{f} \leqslant \lambda f. \sem{\qinfer{\cmd}{f}}$,
%   that is the inference algorithm is sound.
% \end{theorem}
%
\begin{theorem}[Soundness Theorem]
  \label{t:soundness}
  Let $\cmd$ be a program, $f$ be a cost expression, and $\alpha$ be an assignment.
  Then 
  $\sem[\alpha]{\qinfer{\cmd}{f}} \geq \qet{\cmd}{\sem[\alpha]{f}}$, if all side-conditions are met.
\end{theorem}

\newcommand{\GSA}{\TFUN~\ell~\astin}
\newcommand{\GSB}[1][]{\TFUN~\ell~\symbolic{\astate_{#1}}'}
\newcommand{\PSYM}[1]{\symbolic{p_{#1}}}

\begin{example}
    \label{ex:X2a}
    \ms{rewritten}
    We continue with the running example (Figure~\ref{fig:X2}), inspecting the while loop of the ${-}X$ program.
    Assume that the cost of exiting the loop is zero.
    For brevity, let $\oper{U}$ be the composition of $\pw{\q_1 <* X; \q_1 <* H; \q_1,\q_2 <* CNOT}$.
    Then $\qinfer{\mbox{$WHILE{-}X$}}{\TCOST~0} = \GSA$, with the following side-conditions:
% \begin{flalign*}
% % \intertext{For the while loop we return a fresh unknown expectation  and the following constraints. The symbolic state is restricted to $\x, a_{13}, a_{14}$, where $\x$ is the loop counter and $a_{13}, a_{14}$. The second constraint trivially holds for any expectation.}
%     % \phantom{xx} & \qinfer{\mbox{$WHILE{-}X$}}{\TCOST~0} = \GSA && \\
%     & \phantom{x}\neg \x           \vdash \GSA \geq \underline{\qinfer{\pw{\CONSUME~1;\q_1,q_2 <* U;\x^{\Bool} <- \MEAS{\q_1} }}{\GSA}} && \\
%     & \phantom{x}\phantom{\neg} \x \vdash \GSA \geq \TCOST~0
% \end{flalign*}
\begin{align*}
    \neg \x           & \vdash \GSA \geq \qinfer{\pw{\CONSUME~1;\q_1,q_2 <* U;\x^{\Bool} <- \MEAS{\q_1} }}{\GSA} && \\
    \x & \vdash \GSA \geq \TCOST~0
\end{align*}
The latter constraint trivially holds for any expectation.
Our heuristic determines to track only the matrix variables $a_{13}, a_{24}$ in the symbolic state, i.e.\xspace $\astin = \{ \x \mapsto \x, a_{13} \mapsto a_{13}, a_{24} \mapsto a_{24}\}$.
% We define some shorthands for the symbolic probabilities and states to improve readability.
% \begin{align*} 
%     \PSYM{0}' = \sfrac{1}{2} ( 1 + 2a_{13} + 2a_{24}) \quad
%     & \PSYM{1}' = \sfrac{1}{2} ( 1 - 2a_{13} - 2a_{24})  \\
%     \symbolic{\astate_0} = \{ \x \mapsto \x, a_{13} \mapsto a_{13}, a_{24} \mapsto a_{24}\} \quad
%     & \symbolic{\astate_{k}}' = \{ \x \mapsto k, a_{13} \mapsto 0, a_{24} \mapsto 0\}
% \end{align*}
%
Next, the body of the loop is analysed.
% They are not modified after the measurement.
\begin{flalign*}
    \phantom{xx}& \qinfer{\pw{\CONSUME~1;\q_1,\q_2 <* U;\x^{\Bool} <- \MEAS{\q_1} }}{\GSA} \\
    % & = \qinfer{\pw{\CONSUME~1;\q_1 <* X;\q_1,\q_2 <* CNOT;\x^{\Bool} <- \MEAS{\q_1} }}{\GSA} &&\\
    & = \qinfer{\pw{\CONSUME~1;\q_1,\q_2 <* U}}{\TMEASURE~ (d_1\!+\!d_2, \GSB[0])~(d_3\!+\!d_4, \GSB[1])} &&\\
    % &= \qinfer{\pw{\CONSUME~1}} {\TMEASURE~ (p_0 , \GSB[0])~(p_1, \GSB[1])} && \\
    &\Rightarrow 1 + {\TMEASURE~ ( \symbolic{p_0}', \GSB[0])~(\symbolic{p_1}', \GSB[1])} && 
\end{flalign*}
Here, $d_1+d_2$ and $d_3+d_4$ are obtained from the measurement operator applied to $\astin$.
The symbolic states after applying $\oper{U}$ and the measurement operation, as well as the probabilities of reaching these states, are given by:
\begin{align*} 
    \PSYM{0}' = \sfrac{1}{2} ( 1 + 2a_{13} + 2a_{24}) \quad
    & \PSYM{1}' = \sfrac{1}{2} ( 1 - 2a_{13} - 2a_{24})  \\
    \symbolic{\astate_{0}}' = \{ \x \mapsto 0, a_{13} \mapsto 0, a_{24} \mapsto 0\}  \quad
    & \symbolic{\astate_{1}}' = \{ \x \mapsto 1, a_{13} \mapsto 0, a_{24} \mapsto 0\}
\end{align*}
% The probabilities $\PSYM{0}'$ and $\PSYM{1}'$ indicate the probability of measuring $0$ and $1$ on qubit $\q_1$ after applying $\oper{U}$ to the initial symbolic state $\symbolic{\rho_0}$ and are correspond to $\PSYM{0}'$ and $\PSYM{1}'$ in the measurement term.
%     %
Ultimatively, for the upper invariant of the loop body we obtain the following constraint: $\neg \x \vdash \GSA \geq 1 + {\TMEASURE~ ( \symbolic{p_0}', \GSB[0])~(\symbolic{p_1}', \GSB[1])}$.
    % \begin{flalign*}
        % \phantom{xx}& \neg \x \vdash \GSA \geq 1 + {\TMEASURE~ ( p_0, \GSB[0])~(p_1, \GSB[1])} &&\\
% \intertext{Consider the template $\alpha(l) \triangleq \mathsf{c_{13}} \cdot \guarded{ \neg \x \wedge a_{13} \rgeq 0} \cdot a_{13} + \mathsf{c_{24}} \cdot \guarded{ \neg \x \wedge a_{24} \rgeq 0} \cdot a_{24} + \mathsf{c} \cdot \guarded{\neg \x} \cdot 1$ with indeterminate variables $\mathsf{c_{13}}, \mathsf{c_{24}}, \mathsf{c}$.}
% \phantom{xx}& \neg \x \vdash \alpha(\ell)(\symbolic{s_0}) \geq 1 + {\TMEASURE~ ( p_0, \alpha(\ell)(\symbolic{s_0'}))~(p_1, \alpha(\ell)(\symbolic{s_1'}))} && \\
% \intertext{Applying the \textsc{Measure}, and the \textsc{FunR} rule we obtain.}
% \phantom{xx} & \neg \x \wedge \isProb{p_0} \wedge \isProb{p_1} \vdash \alpha(\ell)(\symbolic{s_0}) \geq 1 +  p_0 \cdot (\guarded{\neg \x} \cdot \mathsf{c_{13}}) + p_1 \cdot (\guarded{\neg \x} \cdot \mathsf{c_{24}}) && \\
% &   \neg \x \wedge \isProbOne{p_0} \vdash \alpha(\ell)(\symbolic{s_0}) \geq 1 + p_0 \cdot (\guarded{\neg \x} \cdot \mathsf{c_{13}}) && \\
% &   \neg \x \wedge \isProbOne{p_1} \vdash \alpha(\ell)(\symbolic{s_0}) \geq 1 + p_1 \cdot (\guarded{\neg \x} \cdot \mathsf{c_{24}})
    % \end{flalign*}
% These constraints can be further refined such that they are amenable to the \textsf{SMT} solver.
% The assignment $\mathsf{c_{13}} = \mathsf{c_{24}} = \mathsf{c} = 2$ is a model.
% Therefore, $g_\ell \triangleq 2 \cdot \guarded{ \neg \x \wedge a_{13} \rgeq 0} \cdot a_{13} + 2 \cdot \guarded{ \neg \x \wedge a_{24} \rgeq 0} \cdot a_{24} + 2 \cdot \guarded{\neg \x} \cdot 1$ is an upper bound.
\end{example}

\begin{figure*}[t]
    \centering
    % \hrulefill
    \[
\begin{array}{c}
\infer[\textsc{[FunL]}]{ \phi \vdash \TFUN~\funsym~\symbolic{\astate} \geq f}{\phi \vdash \alpha(\funsym)(\symbolic{\astate}) \geq f}
\qquad
\infer[\textsc{[Cost]}]{ \phi \vdash e \geq \TCOST~\cexp}{\phi \vdash e \geq \cexp}
\qquad
\infer[\textsc{[FunR]}]{ \phi \vdash e \geq \TFUN~\funsym~\symbolic{\astate}}{\phi \vdash e \geq \alpha(\funsym)(\symbolic{\astate})}
\\[2ex]
\infer[\textsc{[Tick]}]{ \phi \vdash e \geq \TCONSUME~\nexp~\texp }
    { \phi \vdash e \geq \guarded{\nexp \igeq 0} \cdot \nexp + \texp }
\qquad
\infer[\textsc{[Cond]}]{ \phi \vdash e \geq \TCOND~\bexp~\texp_1~\texp_2 }
    {\phi \wedge \bexp \vdash e \geq \texp_1  \qquad \phi \wedge \neg \bexp \vdash e \geq \texp_2 }
\\[2ex]
\infer[\textsc{[Meas]}]{ \phi \vdash e \geq \TMEASURE~(\aexp_1, \texp_1)~(\aexp_2, \texp_2) }
    { \begin{array}{c}
    \phi \wedge \isProbOne{\aexp_1} \vdash e \geq \texp_1 \qquad \phi \wedge \isProbOne{\aexp_2} \vdash e \geq \texp_2 \\ 
    \phi \wedge \isProb{\aexp_1} \wedge \isProb{\aexp_2} \vdash e \geq \aexp_1 \cdot \texp_1 + \aexp_2 \cdot \texp_2 
\end{array}
    }
\\[2ex]
\infer[\textsc{[MeasP]}]{ \phi \vdash e \geq \TMEASURE~(\aexp_1, \texp_1)~(\aexp_2, \texp_2) }
    {\phi \implies \aexp_1 \rgt 0 \wedge \aexp_2 \rgt 0 \quad \phi \wedge \isProb{\aexp_1} \wedge \isProb{\aexp_2} \vdash e \geq \aexp_1 \cdot \texp_1 + \aexp_2 \cdot \texp_2}
\end{array}
    \]
    \vspace{-2\belowdisplayskip}
\caption{Inference rules for term constraints to cost constraints.}
\label{fig:term-to-cost}
\end{figure*}

\subsection{Inference of Expectation Cost Constraints}
\label{ss:inference-cost-constraints}

Let $\cmd$ be a program. 
Then $\qinfer{\cmd}{\texp}$ returns a set of constraints $\phi \vdash e \geq f$ over terms.
% In this section we present a set of inference rules to reduce constraints over terms to constraints over cost expressions.
% Figure~\ref{fig:term-simp} provides a couple of useful simplification rules, while
Figure~\ref{fig:term-to-cost} provides a set of inference rules to reduce term constraints to cost constraints.
By construction, the left-hand side of the term constraints is always of the form $\TFUN~\funsym~\astin$,
and besides \textsc{[FunL]}, all other rules eliminate the term structure of the right-hand side such that all obligations eventually form cost constraints.
Rules affecting the right-hand side can be applied in sub-terms. 
For clarity, the term context is omitted in the inference and a meta rule \textsc{[Context]} is added.
Further, the inference rule \textsc{[CostFree]} voids the cost expressions to implement the expected value transformer $\qevt{\cmd}{f}$.
\[
\infer[\textsc{[Context]}]{ \phi \vdash e \geq C[f] }
    { \psi \quad \phi \wedge \phi_i \vdash e \geq C[t_i] }
\qquad
\infer[\textsc{[CostFree]}]{ \phi \vdash e \geq \TCONSUME~\nexp~\cexp~\texp }
    { \phi \vdash \texp }
\]
This set of inference rules transforms term expectation constraints to cost expectation constraints.
Most rules are self-explanatory.
Recall that cost expressions are closed under addition and scalar multiplication (with a non-negative expression).
The inference rule \textsc{[Meas]} performs a case distinction on the probability expression.
The inference rule \textsc{[MeasP]} resolves the case distinction for measurements if we can infer that the probabilities are non-zero.
% This guarantees that the probability expression evaluates to non-zero and the measurement operation and normalisation performed for $\texp_1$ and $\texp_2$ are well-defined.
% \[
% \infer[\textsc{[CostFree]}]{ \phi \vdash e \geq \TCONSUME~\nexp~\cexp~\texp }
%     { \phi \vdash \texp }
% \]

% \begin{example}
%     Let $\alpha(\funsym) = [ \x \wedge -a_{12} \igeq 0 ] \cdot -a_{12}$
% \begin{align*}
% \infer
%   {\TMEASURE~(p_{0,1}, \TFUN~\funsym~\{ \x := 0; a_{12} := 0\})~(p_{1,1}, \TFUN~\funsym~\{ \x := 1; a_{12} := 0\})}
% {
%     \psi_1 \quad \psi_2
%     \qquad
%     \phi \wedge \isProb{p_{0,1}} \wedge \isProb{p_{1,1}} \vdash p_{0,1} \cdot [ \x := 0; a_{12} := 0 ] \cdot -a_{12} + p_{1,1} \cdot [ \x := 1; a_{12} := 0 ] \cdot -a_{12}
% }
% \end{align*}    
% \end{example}

\paragraph{Well-definedness.}
The symbolic state is well-defined.
The case distinction by \textsc{[MeasP]} ensures that the probabilities and the denominators are non-zero.
For cost expressions, we restrict the form of the templates such that the denominator of norms is $1$: $\guarded{\bexp \wedge e/1 \rgeq 0} \cdot e/1$.
Therefore, the instantiation of a norm with a symbolic state is well-defined.
Additionally, cost expressions include scaling by non-negative factors and probabilities. 

\begin{example} 
    \label{ex:X2b} 
    We continue with the running example (Example~\ref{ex:X2a}), inferring cost constraints for the loop body.
    Apply the template
    \[
    \alpha(\ell) \triangleq
    \mathsf{c}_{13}\cdot\guarded{\neg\x\wedge a_{13}\ge 0}\cdot a_{13}
    + \mathsf{c}_{24}\cdot\guarded{\neg\x\wedge a_{24}\ge 0}\cdot a_{24}
    + \mathsf{c}\cdot\guarded{\neg\x}\cdot 1.
    \]
    Apply \textsc{[Meas]} and \textsc{[FunR]}:
    \begin{align*}
        & \neg \x \wedge \isProb{\symbolic{p_0}'} \wedge \isProb{\symbolic{p_1}'} \vdash \alpha(\ell)(\astin) \geq 1 +  \symbolic{p_0}' \cdot (\guarded{\neg \x} \cdot \mathsf{c_{13}}) + \symbolic{p_1}' \cdot (\guarded{\neg \x} \cdot \mathsf{c_{24}}) && \\
&   \neg \x \wedge \isProbOne{\symbolic{p_0}'} \vdash \alpha(\ell)(\astin) \geq 1 + \symbolic{p_0}' \cdot (\guarded{\neg \x} \cdot \mathsf{c_{13}}) && \\
&   \neg \x \wedge \isProbOne{\symbolic{p_1}'} \vdash \alpha(\ell)(\astin) \geq 1 + \symbolic{p_1}' \cdot (\guarded{\neg \x} \cdot \mathsf{c_{24}})
    \end{align*}

\end{example}

\subsection{Inference of Polynomial Constraints}
\label{ss:inference-polynomial-constraints}

% With the inference rules of Figure~\ref{fig:term-to-cost} term constraints are reduced to cost constraints.
Figure~\ref{fig:cost-to-poly} provides a set of inference rules to reduce cost constraints to polynomial constraints.
The rules perform case distinctions on guards.
Additional side constraints are imposed to ensure that the rules for division are well-defined.

\begin{lemma}
  \label{l:inference}
    Let $\phi \vdash e \geq f$ be a generated constraint from $\qinfer{\cdot}{\cdot}$ and let $\Delta$ be a derivation from $\phi \vdash e \geq f$ using the inference rules of Figure~\ref{fig:term-to-cost} and Figure~\ref{fig:cost-to-poly}.    
  Then $\phi \vdash e \geq f$ holds, if all derived inequalities in $\Delta$ hold.
\end{lemma}

\begin{figure*}[t]
    % \hrulefill
    \[
\begin{array}{c}
\infer[\textsc{[GuardL]}]{\phi \vdash g_1 + \guarded{\bexp} \cdot g_2 \geq f}
{
    \phi \wedge \bexp      \vdash g_1 + g_2 \geq f \qquad
    \phi \wedge \neg \bexp \vdash g_1 \geq f
}
\\[2ex]
\infer[\textsc{[GuardR]}]{\phi \vdash g \geq f_1 + \guarded{\bexp} \cdot f_2}
{
    \phi \wedge \bexp      \vdash g \geq f_1 + f_2 \qquad
    \phi \wedge \neg \bexp \vdash g \geq f_1
}
\\[2ex]
\infer[\textsc{[DivL]}]{\phi \vdash g_1 / g_2  \geq f_1}
{
    \phi \vdash g_2 \rgt 0 \qquad
    \phi \vdash g_1 \geq f_1 \cdot g_2
}
\qquad
\infer[\textsc{[DivR]}]{\phi \vdash g_1 \geq f_1 / f_2}
{
    \phi \vdash f_2 \rgt 0 \qquad
    \phi \vdash g_1 \cdot f_2 \geq f_1
}
\end{array}
    \]
    % \vspace{-2mm}
    % \hrulefill
%    \vspace{-2mm}
\vspace{-2\belowdisplayskip}
\caption{Inference rules for cost constraints to polynomial constraints.}
\label{fig:cost-to-poly}
\end{figure*}

\subsection{Inference of Certificate Constraints}
\label{ss:inference-certificate-constraints}

The tool reduces upper invariant constraints to polynomial inequalities.
We use an incomplete characterisation of non-negative polynomials based on Handelmann's Theorem~\cite{Handelman88}---and related work on the Positivstellensatz.
This is a commonly used approach for the synthesis of ranking functions (cf.\xspace \cite{ContejeanMTU05,BEFFG16,AMS20}).
First, inequality constraints are transformed using standard equalities into the following form:
$\bigwedge p_i \geq 0 \wedge \bigwedge p_j > 0 \implies p \geq 0$.
Then, constraints for a certificate of non-negativity of $p$ are generated.
Here, a certificate is a combination of non-negative polynomials using operations that preserve non-negativity.
The tool uses a non-linear combination of non-negative polynomials from the premise, weighted by non-negative coefficients $c_i$ with some chosen maximal degree $n$:
$C_n \triangleq \sum_{i = (i_1,\ldots, i_m)\in [0,n]^m }  c_i \prod p_1^{i_1} \cdots p_m^{i_m}$.
Then, equality of the generated certificate $C_n$ and polynomial $p$ is checked by comparison of coefficients.
These equality constraints can be expressed by quantifier free non-linear formulae in \textsf{QF\_NRA} theory.
By default, \TOOL\ uses the \textsf{Z3} solver\footnote{\url{https://github.com/Z3Prover/z3}}.

\begin{example} 
We continue with the running example (Example~\ref{ex:X2b}).
The constraints can be further refined such that they are amenable to the \textsf{SMT} solver.
The assignment $\mathsf{c_{13}} = \mathsf{c_{24}} = \mathsf{c} = 2$ is a model.
Therefore, $g_\ell \triangleq 2 \cdot \guarded{ \neg \x \wedge a_{13} \rgeq 0} \cdot a_{13} + 2 \cdot \guarded{ \neg \x \wedge a_{24} \rgeq 0} \cdot a_{24} + 2 \cdot \guarded{\neg \x} \cdot 1$ is an upper bound for the while loop.
% We resume the inference.
%    \begin{flalign*} 
%         % \phantom{xx} & \qinfer{-X}{\TCOST~0} && \\
%         % & \phantom{x}= \qinfer{\pw{\q_1 <* H;\q_1,q_2 <* CNOT;\x^{\Bool} <- \MEAS{\q_1};\mbox{$WHILE{-}X$}}}{\TCOST~0} && \\
%         % & \phantom{x}= \qinfer{\pw{\q_1 <* H;\q_1,q_2 <* CNOT;\x^{\Bool} <- \MEAS{\q_1}}}{\qinfer{\mbox{$WHILE{-}X$}}{\TCOST~0}} && \\
%         \phantom{xx} & \qinfer{\pw{\q_1 <* H;\q_1,q_2 <* CNOT;\x^{\Bool} <- \MEAS{\q_1}}}{\qinfer{\mbox{$WHILE{-}X$}}{\TCOST~0}} && \\
%         & \phantom{x}= \qinfer{\pw{\q_1 <* H;\q_1,q_2 <* CNOT;\x^{\Bool} <- \MEAS{\q_1}}}{\TCOST~g_\ell} && \\
%         & \phantom{x}= \qinfer{\pw{\q_1 <* H;\q_1,q_2 <* CNOT}}{\TCOST~(d_1 + d_2) \cdot 2 \cdot 1} && \\
%         & \phantom{x}= \TCOST~(\sfrac{1}{2} + a_{13} + a_{24}) \cdot 2 \cdot 1 &&
%     \end{flalign*}
%

Finally, we obtain the upper bound $(\sfrac{1}{2} + a_{13} + a_{24}) \cdot 2 \cdot 1$ on the expected cost for the ${-}X$ program by symbolic inference.
\end{example}

\subsection{Algorithm}
\label{ss:algorithm}

% We have implemented this approach in a dedicated tool termed \TOOL.
% In this subsection we discuss the main algorithm and some implementation details.
%
% Next, we discuss implementation details for automating the expected cost inference.

%\subsubsection{Automation}

\SetKwFunction{Main}{Main}
\SetKwFunction{SolverGen}{InferWhile}
\SetKwFunction{GetAsserts}{GetAsserts}
\SetKwFunction{Solver}{Infer}
\SetKwFunction{ToPoly}{ToPoly}
\SetKwFunction{ToSmt}{SMT}
\SetKwFunction{ToBound}{ToBound}
\SetKwFunction{TemplateGen}{TemplateGen}
\SetKwFunction{Template}{Template}
\SetKwFunction{Stack}{St}
\SetKwFunction{ToPolyConstraints}{ToPolyConstraints}

% \paragraph{Algorithm.}

\SetAlgoNoEnd
\SetAlgoNoLine
\SetInd{0.3em}{0.3em}
\SetKwProg{Fn}{fun}{:}{}
\SetKwFunction{StackPush}{St.Push}
\SetKwFunction{StackEmpty}{St.Empty}
\SetKwFunction{StackPopAll}{St.PopAll}
\SetKwFunction{Emit}{St.Push}
\SetKwProg{Fn}{fun}{:}{}

\ms{shorter alternative without pseudocode}
Recall Figure~\ref{fig:workflow} depicting the workflow of {\TOOL}. 
The main algorithm leverages the structure of the transformer framework and implements a \emph{modular bottom-up} analysis.
As outlined, the algorithm alternates between symbolic execution to infer constraints for upper invariants and constraint solving propagating inferred cost expressions through the program structure.
Strongly connected components (SCCs) of the control flow graph are analysed separately, assuming that the cost expressions for the continuation are already inferred.
Additionally, the \ref{idents:rext-sep} law is employed to inspect (nested) loops independently.

A key aspect of the approach is the use of \emph{templates} for cost expressions during constraint generation.
Templates are syntactic forms for candidate invariants and are instantiated with unknown coefficients to be solved for.
They are typically constructed from the guards of loops and the cost of the continuation, capturing relevant program variables and symbolic state components.
Standard templates include linear forms, mixed-linear and non-linear forms.
% The choice of template directly impacts both the precision and scalability of the analysis: richer templates can yield tighter bounds but may increase solving complexity.
In the running example, as shown in Example~\ref{ex:X2b}, a linear template is used.
The pseudocode of the algorithm is provided in Appendix~\ref{app:algorithm}.

% GM === inserted case stuy section
\subsection{Evaluation}
%\section{Evaluation}
%\label{s:case_studies}
\ms{expanded a bit}

\begin{table}[t]
\caption{Validation Results on a Benchmark of Mixed Classical-Quantum Programs.}
\label{tab:case-studies}
\begin{tabularx}{\textwidth}{X X >{\raggedleft\arraybackslash}X}
% \toprule
\textbf{Program} & \textbf{Inferred Bound} & \textbf{Time (sec)} \\
\addlinespace[6pt]
\multicolumn{3}{l}{\textbf{examples from~\cite{AMPPZ:LICS:22}}} \\
% \midrule
$COIN{-}TOSSING$               &  $2 + \guarded{-a_{12} \rgeq 0} \cdot 2 \cdot -a_{12}$             & $\leqslant 1s$         \\
$RUS$               &  $\sfrac{8}{3}$             & $\leqslant 1s$         \\
$CHAIN4$                     & $36$              & $\leqslant 20s$         \\
$CHAIN$                      &    $\guarded{k+4 \igeq 0} \cdot 148 \cdot (k + 4)$              &     $\leqslant 30s$    \\
$CHAIN{-}SIMPLE$               &  $\guarded{k+4 \igeq 0} \cdot 148 \cdot (k + 4)$              &    $\leqslant 2s$      \\
$QUANTUM{-}WALK$            &   $?$             &          \\
$QUANTUM{-}WALK(2)$            &   $d_3 + d_4$           &  $\leqslant 1s$        \\
\addlinespace[3pt]
\multicolumn{3}{l}{\textbf{examples from~\cite{ChenCHHLLT25}}}  \\
% \midrule
${-}X$                         & $2 \cdot (\sfrac{1}{2} + a_{13} + a_{24})$                      &    $\leqslant 1s$             \\
$WMGROVER(k)$                 &      $?$                &                 \\
\addlinespace[6pt]
\multicolumn{3}{l}{\textbf{examples from~\cite{PaetznickS14}}} \\
% \midrule
$(I+2iZ)/\sqrt5$            &  $\sfrac{8}{5}$           & $\leqslant 1s$                \\
$(2X+\sqrt2 Y+Z)/\sqrt7$    &  $\sfrac{8}{7}$            & $\leqslant 1s$                \\
$(I+i\sqrt2X)/\sqrt3$       &  $\sfrac{4}{3}$            & $\leqslant 1s$                \\
$(3I+2iZ)/\sqrt13$          &  $\sfrac{32}{13}$           & $\leqslant 1s$                \\
$(4I+iZ)/\sqrt17$           &  $\sfrac{64\sqrt{2}}{51\sqrt{2} + 17}$            & $\leqslant 1s$                \\
$(5I+2iZ)/\sqrt29$          &  $\sfrac{64\sqrt{2}}{29\sqrt{2} + 29}$            & $\leqslant 1s$                 \\
\addlinespace[3pt]
\multicolumn{3}{l}{\textbf{examples from IBM showcase}} \\
% \midrule
% $RUS2$~(Figure~\ref{fig:trial-code})                       & $2 + \sfrac{3}{4}$   &         $\leqslant 1s$        \\
% $RUSWHILE$~(Figure~\ref{fig:ruswhile})                     & $\sfrac{16}{5}$                    &  $\leqslant 1s$               \\
$RUS2$                       & $2 + \sfrac{3}{4}$   &         $\leqslant 1s$        \\
$RUSWHILE$                     & $\sfrac{16}{5}$                    &  $\leqslant 1s$               \\
% \bottomrule
\end{tabularx}
\smallskip
%blocksatz in figure lost
%\begin{flushleft}

\noindent
% The table indicates the upper bounds on the expected cost analysis obtained fully automatically by~\TOOL\ for various mixed classical-quantum programs taken from the literature.
% We also indicate the solving time in seconds. 
The table shows upper bounds on expected cost obtained fully automatically by~\TOOL\ for mixed classical-quantum programs from the literature, with solving times in seconds. 
A question mark $?$ indicates that our implementation was unable to find a bound.
The maximum number of qubits supported in our experiments is $6$, as exemplified by the $CHAIN$ benchmark. 
In case of success, the inferred upper bound matches the tightest known results from the literature.
The experiments were performed on a laptop equipped with an AMD Ryzen 9 with 32~GB RAM

% The table shows upper bounds on expected cost obtained fully automatically by~\TOOL\ for mixed classical-quantum programs from the literature, with solving times in seconds. 
% A question mark $?$ indicates no bound was found.
% The maximum supported number of qubits is $6$ (see $CHAIN$ benchmark). 
% When successful, inferred bounds match the tightest known results.
% Experiments ran on an AMD Ryzen 9 laptop with 32~GB RAM.

%\end{flushleft}
\vspace{-2ex}
\end{table}

% In this section, we showcase a variety of case studies drawn from the literature to demonstrate the capabilities of our prototype implementation, cf.~\ref{tab:case-studies}.
% The source code of selected programs is provided in Appendix~\ref{app:case-studies}.
% If not specified otherwise, we instrument while loops with a $\CONSUME~1$ statement.
% Therefore, the expected cost represents the expected runtime of the program---ignoring the constant factor from the number of operations.

We evaluate~\TOOL\ on examples from the literature to demonstrate its capabilities (cf.~Table~\ref{tab:case-studies}).
The motivating examples for the quantum expectation transformer from~\cite{AMPPZ:LICS:22}, previously analysed manually, are now handled automatically.
Repeat-until-success circuits~\cite{PaetznickS14} serve as a representative application of quantum programs with mid-circuit measurements and classical control flow.
We also include benchmark examples from \textsf{AutoQ 2.0}~\cite{ChenCHHLLT25} and the IBM showcase.
Source code is provided in Appendix~\ref{app:case-studies}.
We instrument while loops with $\CONSUME~1$ to represent expected runtime.
\TOOL\ provides sound upper bounds, which we compare against established bounds from the literature.
% We discuss the motivating examples from~\cite{AMPPZ:LICS:22,AMPP24}.
% %
% The $COIN{-}TOSSING$ program, depicted in Figure~\ref{fig:coin-tossing}, is a single-qubit program performing coin tossing by repeatedly applying a Hadamard gate and measuring the result.
% For this program, the expected cost depends on the initial state of the qubit.
% %
% The $RUS$ example is one of the repeat-until-success examples we discuss below.
% %
The $CHAIN$ program models a chain of $k$ entangled qubits where $\pw{CZ}$ gates may fail.
It uses nested and sequential loops, and among evaluated programs has the most qubits ($6$) and measurements ($22$).
$QUANTUM{-}WALK$ simulates a quantum walk on a line with $n$ qubits but cannot be analysed as it operates on a parametrised state space with unsupported unitaries. 
We specialise it to two qubits in $QUANTUM{-}WALK(2)$, yielding a precise bound.
$WMGROVER(k)$ is a variation of Grover's search using weak measurements (cf.~\cite{MartinezH22}) on $k$ qubits.

\section{Conclusion and Future Work}
\label{s:conclusion}

In this work, we have presented the first---to the best of our knowledge---imple\-ment\-ation of a recently proposed weakest pre-expectation semantics for mixed classical-quantum programs, cf.~\cite{AMPPZ:LICS:22,AMPP24}.
This transformer amounts to a formal framework for the analysis of runtime and costs of quantum programs in expectation.
Our tool~$\TOOL$ is able to fully automatically infer precise upper bounds on the expected cost of quantum programs, like \emph{repeat-until-success} circuits
that feature classical control flow as well as non-trivial features like mid-circuit measurements.

We provide ample experimental evidence of the viability of the prototype implementation.
%$\TOOL$ is evaluated through a number of case studies taken from the literature, as well as online resources.
We emphasise that all quantum programs in our benchmark require a high-level
encoding as a mixed classical-quantum programming language and are not expressible as sole quantum circuits. For almost all cases our tool can fully automatically
derive an expected cost analysis that, so far, could only be obtained manually.

Arguably the biggest challenge to overcome in a verification framework for quantum programs is scalability and improved viability. In our opinion, the best way to tackle this challenge lies in
\begin{inparaenum}[(i)]
  \item the development of suitable abstract state representations and 
  \item restricting the program to specific classes of high-level quantum programs. 
\end{inparaenum}
This will be subject to future work.

\bibliographystyle{splncs04}
\FloatBarrier
%GM: strings not defined
\bibliography{references}

\clearpage
\appendix
\section{Appendix}
\label{s:appendix}

\subsection{Syntax}
\label{app:syntax}

The syntax grammar for the mixed classical-quantum imperative programming language \IMQ\ is given in Figure~\ref{app:fig:synt}.
\begin{figure}[ht]
\vspace{-2\belowdisplayskip}
\[
  \begin{array}{rlll}
    \NExp &  \ni \nexp, \nexp_1, \nexp_2
    & \rgl &   \x^{\NVar} \mid i \in \mathbb{Z} \mid
             \nexp_1 + \nexp_2   \mid  \nexp_1 - \nexp_2  \mid \nexp_1 \times \nexp_2 \\
    \BExp &  \ni \bexp, \bexp_1, \bexp_2
    & \rgl &  \x^{\Bool}   \mid \true \mid \false \mid
             \nexp_1 = \nexp_2  \mid  \nexp_1 < \nexp_2 \mid \neg \bexp \mid \bexp_1 \wedge \bexp_2 \mid \bexp_1 \vee \bexp_2 \\
    \Exp &  \ni \e, \e_1, \e_2
    & \rgl &  \nexp \mid \bexp \\
    \Cmds &  \ni  \cmd, \cmd_1, \cmd_2
    & \rgl &  \SKIP \mid  \x^{\K} <- \e^{\K}  \mid  \cmd_1; \cmd_2\mid\  \!  \IF \bexp^{\Bool} \THEN \cmd_1  \ELSE \cmd_2   \\
          & & &\mid\ \WHILE \bexp^{\Bool} \DO \cmd
                \mid  \qs^{\Qubits}  <* \ope \mid   \x^{\Bool} <- \MEAS{\q^{\Qubits}} \mid \CONSUME~\nexp
  \end{array}
\]
\vspace{-2\belowdisplayskip}
 \caption{Syntax of the mixed classical-quantum imperative programming language \IMQ.}
\label{app:fig:synt}
\end{figure}

\subsection{Case Study: Repeat Until Success}
\label{app:repeat-until-success}

Here, we discuss one implementation of repeat-until-success programs~\cite{PaetznickS14} showcased on IBM Quantum\textsuperscript{\textregistered{}} devices\footnote{\url{https://learning.quantum.ibm.com/tutorial/repeat-until-success}, last accessed January 2026.}.
This quantum program has three qubits, two control qubits $\cq_1$ and $\cq_2$ and one target qubit $\trgt$.
The classical variables $\x_1$ and $\x_2$ are used to store and process the measurement outcomes of the control qubits, while the classical variable $\notdone$ is used to propagate the success of the operation.
On a successful measurement outcome of the control qubits---$\cq_1$ and $\cq_2$ are measured to be in state $\ket{00}$---%
the rotation operation $\pw{R_X}(\theta)$, with $\cos \theta = \sfrac{3}{5}$, has been successfully performed on the target qubit.
\begin{figure}[ht]
\centering
\scalebox{0.90}{
\begin{quantikz}
\lstick{$\cq_1$}    & \gate{\ket{0}} \slice{\pw{init}}  & \gate{H} & \ctrl{1} & \qw      & \ctrl{1} & \gate{H} \slice{\pw{U}} & \meter{}    &         & \slice{\pw{measure}} &  & \slice{\pw{recover}}   & \\
\lstick{$\cq_2$}    & \gate{\ket{0}}   & \gate{H} & \ctrl{1} & \qw      & \ctrl{1} & \gate{H} &             & \meter{}    & &    & & \\
\lstick{$\trgt$}    &                  & \gate{H} & \targ{}  & \gate{S} & \targ{}  & \gate{H} & \qw         &             & & & \gate[wires=1]{X} & \qw \\
\lstick{$\x_1$}     & \cw              & \cw      & \cw      & \cw      & \cw      & \cw      & \cwbend{-3} & \cw         & \cwbend{2} & \cw &  \cw & \cw\\
  \lstick{$\x_2$}     & \cw              & \cw      & \cw      & \cw      & \cw      & \cw      & \cw         & \cwbend{-3} & \cwbend{1} & \cw & \cw & \cw \\
\lstick{$\notdone$} & \cw              & \cw      & \cw      & \cw      & \cw      & \cw      & \cw         & \cw         & \gate[style={shape=circle}]{\ensuremath{\vee}} \cw & \cw & \cwbend{-3} & \cw
\end{quantikz}%
}%
\vspace{-10pt}
\caption{The quantum circuit for the $TRIAL$ subprogram.}
\label{app:fig:trial-circuit}
\end{figure}
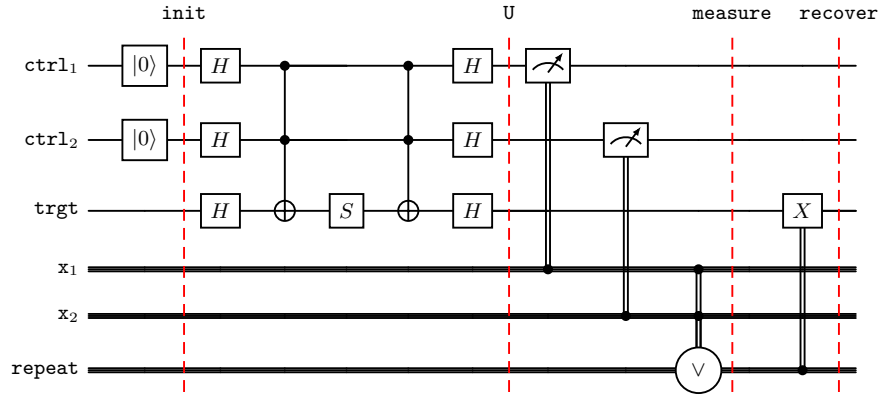

\begin{figure}[ht] 
  \centering
  \begin{minipage}[c]{0.47\textwidth} $ TRIAL(\trgt^{\Qubits}, \notdone^{\Bool}) \triangleq \\ \phantom{xx}\pw{ \cq_1 = \ket{0};\tikzrom{reset1}\\ \cq_2 = \ket{0};\tikzrom{reset2}\\ \cq_1 <* H;\tikzrom{u1}\\ \cq_2 <* H;\\ \trgt <* H;\\ \cq_1,\cq_2,\trgt <* CCNOT;\tikzrom{u3}\\ \trgt <* S;\\ \cq_1,\cq_2,\trgt <* CCNOT;\\
      \trgt <* H;\\
      \cq_2 <* H;\\
      \cq_1 <* H;\tikzrom{u2}\\

      \x_1 <- \MEAS{\cq_1}; \tikzrom{m1}\\
      \x_2 <- \MEAS{\cq_2};\\
      \notdone = \x_1 \vee \x_2; \tikzrom{m2}\\

      \tikzrom{p1}
      \IF \notdone \THEN { 
        \trgt <* X;\phantom{more } \tikzrom{p3}
      }                   \tikzrom{p2}
    }
    $
    \AddNotecopy{reset1}{reset2}{reset1}{$\triangleq \pw{init}$}
    \AddNotecopy{u1}{u2}{u3}{$\triangleq \pw{U}$}
    \AddNotecopy{m1}{m2}{m2}{$\triangleq \pw{measure}$}
    \AddNotecopy{p1}{p2}{p3}{$\triangleq \pw{recovery}$}
\end{minipage}
\
\begin{minipage}[c]{0.47\textwidth}
    $
    RUS2(\trgt^{\Qubits}) \triangleq \\
    \phantom{xx}\pw{
      \notdone^{\Bool} = \true;\\

      \IF \notdone \THEN {
        \CONSUME~2;\\
\mbox{$TRIAL$}(\trgt, \notdone);\\
      }\\
      \IF \notdone \THEN {
        \CONSUME~2;\\
\mbox{$TRIAL$}(\trgt, \notdone);\\
      }
    }
    $
 \end{minipage}
\caption{The repeat-until-success (RUS) program of the IBM showcase.}
\label{app:fig:trial-code}
\end{figure}

Figure~\ref{app:fig:trial-circuit} depicts the quantum circuit of the $TRIAL$ operation.
Figure~\ref{app:fig:trial-code} (left) shows the corresponding program in our programming language.
This is the operation that is repeated in the repeat-until-success algorithm until a successful measurement outcome is achieved. Note that the $TRIAL$ operation is meant as a subprogram and not an isolated part in which the original state is prepared externally for repeated executions.
The program is split up in four parts, initialisation (\pw{init}), application of a sequence of quantum gates (\pw{U}), measurement (\pw{measure}) and recovery (\pw{recover}).
First, the two control qubits $\cq_1$ and $\cq_2$ are initialised to $\ket{0}$.
Then, a series of quantum gates are applied to the control and target qubits, including the Hadamard gate $\pw{H}$, the controlled-controlled-NOT (or Toffoli) gate $\pw{CCNOT}$ and the phase shift gate $\pw{S}$.
This is followed by the measurement of the control qubits.
A measurement returns the classical outcome $0$ or $1$.
We interpret the measurement outcomes as Boolean values, with $0$ representing the Boolean value false and $1$ representing the Boolean value true.
If the measurement of both control qubits return false then the rotation has been successfully performed on the target qubit, and we set the variable $\notdone$ to false.
Otherwise, a conditional recovery operation (application of the gate $\pw{X}$) is performed that transforms the target qubit back to its initial state.
The program $RUS2$ of Figure~\ref{app:fig:ruswhile} depicts the complete
% repeat-until-success program
IBM showcase. It performs the $TRIAL$ operation a fixed number of times, in this case two times.
Notice that, if the trial is successful once, then repeated applications correspond to a no-op.
Given enough tries, the program is likely to perform the desired operation on the target qubit.
However, it does not guarantee that the rotation will be executed.

We are interested in \emph{quantitative properties} of programs, which instantiated for
the IBM showcase, would be the expected number of \pw{T}~gate applications in a successful run of the algorithm. 
To make this concrete we instrument the code with a \emph{ticking operation} $\CONSUME$ that counts $\pw{T}$~gates. 

Typically, one would express a conditional repetition with a loop statement.
% However, at the time of writing the IBM Quantum\textregistered{} platform only supports classical conditional control flow that may depend on mid-circuit measurements but not while statements.
%
The program $RUSWHILE$ of Figure~\ref{app:fig:ruswhile} depicts the repeat-until-success program using a classical while statement to control the repetition of the $TRIAL$ operation.
We augment this example with the \pw{tick} $2$ statement assuming that we execute the phase shift gate $\pw{S}$ with $\pw{T^2}$.
The classical control flow is easy to understand, however, reasoning about the sequence of quantum operations and the result of the measurement operations is non-trivial.

\begin{figure}[ht]
  \centering
\begin{minipage}[t]{0.48\textwidth}
    $
    RUS2(\trgt^{\Qubits}) \triangleq \\
    \phantom{xx}\pw{
      \notdone^{\Bool} = \true;\\

      \IF \notdone \THEN \mbox{$TRIAL$}(\trgt, \notdone);\\
      \IF \notdone \THEN \mbox{$TRIAL$}(\trgt, \notdone);\\
    }
    $
 \end{minipage}
\
\begin{minipage}[t]{0.40\textwidth}
    $
    RUSWHILE(\trgt^{\Qubits}) \triangleq \\
    \phantom{xx}\pw{
      \notdone^{\Bool} = \true;\\
      \WHILE \notdone \DO {
        \CONSUME~2;\\
        \mbox{$TRIAL$}(\trgt, \notdone);
      }
    }
    $
 \end{minipage}
\caption{The repeat-until-success (RUS) program of the IBM showcase.}
\label{app:fig:ruswhile}
\end{figure}

We provide some initial informal insights into the challenges of reasoning about quantum programs.
The state $\ket{\phi}$ after the initialisation $\textsf{init}$ is $\ket{00} \otimes (\alpha \ket{0} + \beta \ket{1})$ with $\left|\alpha\right|^2 + \left|\beta\right|^2 = 1$.
The sequence of elementary operations in $\textsf{U}$ can be expressed by a single unitary matrix $U$. 
The application of the unitary $U$ to quantum state $\ket{\phi}$, $U \cdot \ket{\phi}$,
yields the state before the measurement operations in $\textsf{measure}$. Recall that~$i$ denotes the imaginary number.
{\small\begin{equation*}
  % U \begin{pmatrix}\alpha\\ \beta \\ 0 \\ 0 \\ 0 \\ 0 \\ 0 \\ 0\end{pmatrix} =
  \frac{1}{4}
  \underbrace{\begin{pmatrix}
    2 + 2 i & 1 - i & 0 & 1 - i & 0 & 1 - i & 0 & -1 + i\\
    1 - i & 2 + 2 i & 1 - i & 0 & 1 - i & 0 & -1 + i & 0\\
    0 & 1 - i & 2 + 2 i & 1 - i & 0 & -1 + i & 0 & 1 - i\\
    1 - i & 0 & 1 - i & 2 + 2 i & -1 + i & 0 & 1 - i & 0\\
    0 & 1 - i & 0 & -1 + i & 2 + 2 i & 1 - i & 0 & 1 - i\\
    1 - i & 0 & -1 + i & 0 & 1 - i & 2 + 2 i & 1 - i & 0\\
    0 & -1 + i & 0 & 1 - i & 0 & 1 - i & 2 + 2 i & 1 - i\\
    -1 + i & 0 & 1 - i & 0 & 1 - i & 0 & 1 - i & 2 + 2 i
  \end{pmatrix}}_{{} \triangleq U} 
  \quad
  \underbrace{
  \begin{pmatrix}\alpha\\ \beta \\ 0 \\ 0 \\ 0 \\ 0 \\ 0 \\ 0\end{pmatrix}%
}_{{} \triangleq \ket{\phi}}
  %   =
  % \frac{1}{4}
  % \begin{pmatrix}
  %   \left(1 - i\right) \left(2 i \alpha + \beta\right) \\ 
  %   \left(1 - i\right) \left(\alpha + 2 i \beta\right) \\ 
  %   \beta \left(1 - i\right) \\ 
  %   \alpha \left(1 - i\right) \\ 
  %   \beta \left(1 - i\right) \\ 
  %   \alpha \left(1 - i\right) \\ - 
  %   \beta \left(1 - i\right) \\ - 
  %   \alpha \left(1 - i\right)
  % \end{pmatrix}
\end{equation*}}
Then the application of the unitary $U\ket{\phi}$ yields the state before the measurement operations in $\textsf{measure}$---in vector notation:
\[
U \begin{pmatrix}\alpha\\ \beta \\ 0 \\ 0 \\ 0 \\ 0 \\ 0 \\ 0\end{pmatrix} =
  \frac{1}{4}\begin{pmatrix}
    \left(1 - i\right) \left(2 i \alpha + \beta\right) \\ 
    \left(1 - i\right) \left(\alpha + 2 i \beta\right) \\ 
    \beta \left(1 - i\right) \\ 
    \alpha \left(1 - i\right) \\ 
    \beta \left(1 - i\right) \\ 
    \alpha \left(1 - i\right) \\ - 
    \beta \left(1 - i\right) \\ - 
    \alpha \left(1 - i\right)
  \end{pmatrix}
\]

To reason about the classical control flow, it is necessary to inspect the quantum state and compute the probabilities of the measurement outcomes.
The probability of $P[\cq_1 = 0~\text{and}~\cq_2 = 0]$ is given by the sum of the squared moduli of the first two vector entries:
\begin{equation*}
  P[\cq_1 = 0~\text{and}~\cq_2 = 0] 
  % & = \left\lVert \begin{pmatrix}\frac{\left(1 - i\right) \left(2 i \alpha + \beta\right)}{4} \\ \frac{\left(1 - i\right) \left(\alpha + 2 i \beta\right)}{4} \\ 0 \\ 0 \\ 0 \\ 0 \\ 0 \\ 0\end{pmatrix}\right\rVert^2 
  = \frac{\left|2 i \alpha + \beta\right|^2 + \left|\alpha + 2 i \beta\right|^2}{8} 
  = \frac{5(\left|\alpha\right|^2 + \left|\beta\right|^2)}{8}= \frac{5}{8}
  \tpkt
\end{equation*}

We observe that even for small programs, reasoning about program properties by hand becomes tedious and is prone to errors.
The quantum state space grows exponentially with the number of qubits, and elementary quantum operations---including application of unitaries and measurements---may affect the state in non-trivial ways.
This requires careful computations with amplitudes and probabilities. 
As a result, manual analysis quickly becomes infeasible for realistic programs. 
This motivates the need for computer-aided tools to automate the analysis and verification of quantum programs.

\paragraph{Quantum expectation transformer}

The \emph{quantum expectation transformer} presented by \cite{AMPPZ:LICS:22,AMPP24} is a formal framework for reasoning about expectations of quantum programs and constitutes our basis for automating the expected cost analysis. This quantum expectation transformer is a mapping from expectations---real-valued functions on program states---to expectations in a continuation-passing style:
   $\wpt{\cdot}{f} : \Cmds \to (\AState \to \Rext) \to (\AState \to \Rext)$.
With this framework, it is possible to reason about different program properties including invariants and resource consumption.
This depends on how cost annotations (i.e.\xspace $\CONSUME~\nexp$ statements) are interpreted and how continuations $f$ are defined.

First, let us assume that cost annotations are ignored, that is $\wpt{\CONSUME~\nexp}{f}(\sigma) \triangleq f(\sigma)$,
then $\wpt{\cmd}{f}$ is a function that maps input states $\sigma$ to the expected value of $f$ applied to the post states after executing $\cmd$.
In particular, if $f$ is a predicate that maps a state to $\{0,1\}$, 
then $\wpt{\cmd}{f}(\sigma)$ yields the probability that predicate $f$ holds after executing program $\cmd$ with input state $\sigma$.

If we inspect the expected cost and set $\wpt{\CONSUME~\nexp}{f}(\sigma) \triangleq \lambda \sigma . \max(0, \nexp \sigma) + f(\sigma)$,
then $\wpt{\cmd}{f}(\sigma)$ yields the expected cost of executing $\cmd$ with input state $\sigma$ with respect to the continuation $f$.
In particular, if $\cmd$ is the complete program,
then $\wpt{\cmd}{\lambda \_.\ 0}(\sigma)$ yields the expected cost of executing $\cmd$ with input state $\sigma$---the cost of the remaining execution is zero when reaching a final state.
We highlight that the expected cost analysis encapsulates the \emph{expected runtime} analysis, by preceding each statement with a $\CONSUME~1$ statement.

For automation, we differentiate between loop-free code and code with (possible nested) loops:
\begin{inparaenum}[1)]
    \item For (sub)programs without loops, we use \emph{symbolic evaluation} to mechanically compute the pre-expectations for a given continuation $f$.
    This process handles assignments, quantum gate applications, classical conditional branches, and probabilistic case analyses arising from mid-circuit measurements.
    \item For (sub)programs with loops, we generate \emph{upper-invariant constraints} by introducing unknown expectation functions at control flow points such as loop headers. 
    Symbolic evaluation is then used to relate these functions at control flow points, resulting in a system of constraints.
\end{inparaenum}
Re-consider the program $RUSWHILE$ of Figure~\ref{app:fig:ruswhile}.
We inspect the loop body:
\begin{align*}
  & \wpt{\CONSUME~2; TRIAL(\trgt, \notdone )}{f} \\
  & \qquad = \lambda \sigma .\ 2 + p_{00}(\sigma[\pw{init;U}])\cdot f(\sigma[\pw{init;U;measure}])  \\
  & \qquad \phantom{ xx \lambda \sigma .\ 1} + (1 - p_{00}(\sigma[\pw{init;U}])) \cdot f(\sigma[\pw{init;U;measure;recover}]) \\
  & \qquad = \lambda \sigma .\ 2 + \sfrac{5}{8}\cdot f(\sigma[\pw{init;U;measure}])  
   + \sfrac{3}{8} \cdot f(\sigma[\pw{init;U;measure;recover}])
\end{align*}
% For succinctness we use consider the initial program parts of $TRIAL$: $\pw{init;unitary;measure;undo}$.
Let $\sigma$ denote the state before execution.
We use $\sigma[\cmd]$ to indicate the state after executing $\cmd$.
Here, $p_{00}(\sigma[\pw{init;U}])$ indicates the probability of returning $\ket{00}$ from the measurement.
We obtain the cost expectation constant $2$, which is due to the $\CONSUME~2$ statement, plus the weighted sum of $f$ applied to all possible post states.
With probability $p_{00}(\sigma[\pw{init;U}]) = \sfrac{5}{8}$ the post state is $\sigma[\pw{reset;U;measure}]$, and
with probability $1-p_{00}(\sigma[\pw{init;U}]) = \sfrac{3}{8}$ the post state is $\sigma[\pw{reset;U;measure;recover}]$, which includes the recovery operation to reset the target qubit.
Our tool derives the constant success probability $p_{00}(\sigma[\pw{init;U}]) = \sfrac{5}{8}$ fully automatically. 

% For automation we propose a \emph{template-based} approach (cf.~\cite{CMTU05,BEFFG16,WFGCQS:PLDI:19,AMS20}) to synthesise expectation functions.
% First, \emph{upper-invariant} constraints are formed from while loops $\wpt{\WHILE \bexp \DO \cmd}{F}$:
% \begin{align*}
% \bexp & \vdash G \geq \wpt{\cmd}{G} \\
% \neg \bexp & \vdash G \geq F
% \end{align*}
% Here, $G$ and $F$ constitutes unknown expectation functions, which have yet to be determined.
Let $G$ and $F$ be unknown expectation functions.
Then, we form the upper invariant constraints for $RUSWHILE$:
\begin{align*}
 \notdone & \vdash G \geq \wpt{\CONSUME~2; TRIAL(\trgt, \notdone)}{G} \\
  \neg \notdone & \vdash G \geq F
\end{align*}
To obtain the expected cost of our motivating example, we need to find expectations for $G$ and $F$ that satisfy these constraints for all possible input states $\sigma$.

For automation, we use a \emph{template-based} approach (cf.~\cite{CMTU05,BEFFG16,WFGCQS:PLDI:19,AMS20}) to synthesise expectation functions.
Following~\cite{AMS20,AMS23,pubs}, we introduce a class of \emph{cost-expressions} to form certificate constraints over program variables
and unknown indeterminate constants.
For the motivating example, it suffices to instantiate the constant function $\lambda \sigma .\ \guarded{\notdone} \cdot \mathsf{c_0}$ for~$G$ and the zero function
$\lambda \sigma .\ 0$ for~$F$. Note that the cost of the execution after exiting the loop is zero.
Here, $\guarded{\notdone} \cdot \mathsf{c_0}$ is a guarded expression that evaluates to the unknown indeterminate $\mathsf{c_0}$, if $\notdone$ holds and to $0$ otherwise.
% \begin{align*}
%  \notdone & \vdash \mathsf{c_0} \geq 1 + \sfrac{5}{8} \cdot (\guarded{\notdone} \cdot \mathsf{c_0}) + \sfrac{3}{8} \cdot (\guarded{\notdone} \cdot \mathsf{c_0})\\
%  \neg \notdone & \vdash \mathsf{c_0} \geq 0
% \end{align*}
%
The variable $\notdone$ evaluates to true only for the probabilistic case $\sfrac{3}{8} \cdot f(\sigma[\pw{init;U;measure;recover}])$.
Then, we obtain the following constraints:
\begin{align*}
 \notdone & \vdash \mathsf{c_0} \geq 2 + \sfrac{3}{8} \cdot (\guarded{\notdone} \cdot \mathsf{c_0})\\
 \neg \notdone & \vdash \mathsf{c_0} \geq 0
\end{align*}
By case analysis, these constraints can be reduced further to constraints which are amenable to \textsf{SMT} solvers.
The assignment $\mathsf{c_0} = \sfrac{16}{5}$ satisfies the constraints.
Therefore, $\guarded{\notdone} \cdot \mathsf{c_0}$ is an upper bound on the cost of the loop statement.
Finally, 
\[
  \qet{RUSWHILE}{\lambda \sigma .\ 0} = \sfrac{16}{5} \tpkt
\]

In this example the inferred final constraints may appear rather simple.
However, this is only due to precise handling and simplifications of symbolic expressions to obtain constant probabilities.
Furthermore, a suitable template has been chosen carefully, which simplifies the right-hand side of our constraints significantly.
% Making sure that these pieces fall in line smoothly, was a main challenge for our automation, cf.~Section~\ref{s:automation}.
%
% Our prototype $\TOOL$ yields a tight bound on $RUSWHILE$ (almost) in the blink of an eye.

\subsection{Case Study: {-X}}

Here, we discuss the expected cost analysis of the $-X$ program from Figure~\ref{app:fig:X2}.
This program is a two-qubit system.
It applies a non-standard gate $\pw{-X}$ on the second qubit, i.e.\xspace it applies $\pw{X}$ on the second qubit and negates the amplitude.
For example, $\alpha \ket{10} + \beta \ket{11}$ is transformed into $-\beta \ket{10} - \alpha \ket{11}$.
\begin{figure*}[ht]
  \centering
\begin{minipage}[t]{0.40\textwidth}
    $
    WHILE{-}X(\x^{\Bool},\q_1^{\Qubits}, \q_2^{\Qubits}) \quad\triangleq\quad \\
    \phantom{xx}\pw{
      \WHILE \neg \x \DO {
        \CONSUME~1; \\
        \q_1 <* X; \\
        \tikzrom{cq-mark1}
        \q_1 <* H; \\
        \q_1,q_2 <* CNOT; \hspace{1ex}\tikzrom{cq-mark3} \\
        \x <- \MEAS{\q_1}; 
        \tikzrom{cq-mark2}
      }
    }
    $
\AddNotecopy{cq-mark1}{cq-mark2}{cq-mark3}{$\triangleq \pw{U}$}
 \end{minipage}
\
  \begin{minipage}[t]{0.40\textwidth}
    $
    \xtwo(\q_1^{\Qubits},\q_2^{\Qubits}) \quad\triangleq\quad \\
    \phantom{xx}\pw{
        \tikzrom{cq-mark1}
      \q_1 <* H; \\
      \q_1,q_2 <* CNOT; \\
      \x^{\Bool} <- \MEAS{\q_1};\hspace{1ex}\tikzrom{cq-mark3}
      \tikzrom{cq-mark2} \\
      \mbox{$WHILE{-}X$}(\x, \q_1, \q_2);\\
    }
    $
\AddNotecopy{cq-mark1}{cq-mark2}{cq-mark3}{$\triangleq \pw{U}$}
\end{minipage}
\\[5pt]
\caption{The $-X$ quantum program.}
\label{app:fig:X2}
\end{figure*}

Wlog., we denote the quantum state in Dirac notation instead of density operators for clarity.
As input, we consider parametrised quantum states of the form $\alpha\ket{10} + \beta\ket{11}$, where $\alpha,\beta \in \Comp$ and $|\alpha|^2 + |\beta|^2 = 1$.
We illustrate the expected cost by inspecting the probability distribution induced by operational semantics of \IMQ. 
\begin{align*}
   & \ecost_{-X}(\cst, \alpha\ket{10} + \beta\ket{11}) = \\ 
   & \phantom{x} \lmulti 1:(\pw{\q_1 <* H;\q_1,q_2 <* CNOT;\x^{\Bool} <- \MEAS{\q_1};\mbox{$WHILE{-}X$}}, s,\alpha\ket{10} + \beta\ket{11}) \rmulti  \\
   & \phantom{x}\toomqw{0} \lmulti 1:(\pw{\q_1,q_2 <* CNOT;\x^{\Bool} <- \MEAS{\q_1};\mbox{$WHILE{-}X$}}, s,\alpha\ket{00} + \beta\ket{01} - \alpha\ket{10} - \beta\ket{11}) \rmulti  \\
   & \phantom{x}\toomqw{0} \lmulti 1:(\pw{\x^{\Bool} <- \MEAS{\q_1};\mbox{$WHILE-{X}$}}, s,\alpha\ket{00} + \beta\ket{01} - \alpha\ket{11} - \beta\ket{10}) \rmulti  \\
   & \phantom{x}\toomqw{0} \lmulti 
    \sfrac{1}{2}:(\downarrow, [x := 1],-\alpha\ket{11} - \beta\ket{10}), \\
   & \phantom{\toomqw{0}xxx}\; \sfrac{1}{2}:(\pw{\CONSUME~1; \q_1 <* X; U;\mbox{$WHILE{-}X$}},\ [x := 0],\ \alpha\ket{00} + \beta\ket{01}) \rmulti  \\
   & \phantom{x}\toomqw{\sfrac{1}{2}} \lmulti 
    \sfrac{1}{2}:(\downarrow, [x := 1],-\alpha\ket{11} - \beta\ket{10}), \\
   & \phantom{\toomqw{0}xxx}\; \sfrac{1}{2}:(\pw{\q_1 <* X; U;\mbox{$WHILE{-}X$}},\ [x := 0],\ \alpha\ket{00} + \beta\ket{01}) \rmulti  \\
   & \phantom{x}\toomqw{0} \lmulti 
    \sfrac{1}{2}:(\downarrow, [x := 1],-\alpha\ket{11} - \beta\ket{10}), \\
   & \phantom{\toomqw{0}xxx}\; \sfrac{1}{2}:(\pw{U;\mbox{$WHILE{-}X$}},\ [x := 0],\ \alpha\ket{10} + \beta\ket{11}) \rmulti  \\
   & \phantom{\toomqw{0}xxx}... \\
   & \phantom{x}\toomqw{\sfrac{1}{4}}^\ast \lmulti 
    \sfrac{3}{4}:(\downarrow, [x := 1],-\alpha\ket{11} - \beta\ket{10}) ,
   \sfrac{1}{4}:(\pw{\mbox{$WHILE{-}X$}},\ [x := 0],\ \alpha\ket{00} + \beta\ket{01}) \rmulti  \\
\end{align*}
We observe that, before the measurement $\pw{\x^{\Bool} <- \MEAS{\q_1}}$ the quantum state is always in
$\alpha\ket{00} + \beta\ket{01} - \alpha\ket{10} - \beta\ket{11}$.
The measurement of $\q_1$ yields $\ket{0}$ with probability $\sfrac{1}{2}(|\alpha|^2 + |\beta|^2) = \sfrac{1}{2}$, and thus yields $\ket{1}$ with probability $\sfrac{1}{2}$. We obtain
\begin{equation*}
  \ecost_{-X}(\cst, \alpha\ket{10} + \beta\ket{11}) = \sfrac{1}{2} + \sfrac{1}{4} + \cdots = \sum_{n=1}^{\infty} \sfrac{1}{2^n} = 1
  \tpkt
\end{equation*}
The more general state $\alpha\ket{00} + \beta\ket{01} + \gamma\ket{10} + \delta\ket{11}$ is more complicated to analyse.
The first measurement of $\q_1$ yields $\ket{0}$ with probability $\sfrac{1}{2}(|\alpha+\gamma|^2 + |\beta+\delta|^2)$ and $\ket{1}$ with probability $\sfrac{1}{2}(|\alpha-\gamma|^2 + |\beta-\delta|^2)$.
% We examine this case in more detail when discussing automation.

Next, we present the inference of the expected cost of $\qinfer{-X}{\TCOST~0}$.
To highlight the next sub-term to be evaluated, we use underlining.
First, we fix some expressions to improve readability.
\begin{align*} 
    \PSYM{0}' = \sfrac{1}{2} ( 1 + 2a_{13} + 2a_{14}) \quad
    & \PSYM{1}' = \sfrac{1}{2} ( 1 - 2a_{13} - 2a_{14})  \\
    \astin = \{ \x \mapsto \x, a_{13} \mapsto a_{13}, a_{14} \mapsto a_{14}\} \quad
    & \symbolic{\sigma_{k}}' = \{ \x \mapsto k, a_{13} \mapsto 0, a_{14} \mapsto 0\}
\end{align*}
The expected cost of $-X$ is $\qinfer{-X}{\TCOST~0}$. We proceed bottom-up.
\begin{flalign*} 
    \phantom{xx} & \qinfer{-X}{\TCOST~0} && \\
    & = \qinfer{\pw{\q_1 <* H;\q_1,q_2 <* CNOT;\x^{\Bool} <- \MEAS{\q_1};\mbox{$WHILE{-}X$}}}{\TCOST~0} && \\
    & = \qinfer{\pw{\q_1 <* H;\q_1,q_2 <* CNOT;\x^{\Bool} <- \MEAS{\q_1}}}{\underline{\qinfer{\mbox{$WHILE{-}X$}}{\TCOST~0}}} &&
\end{flalign*}
For the while loop we return a fresh unknown expectation  and the following constraints. The symbolic state is restricted to $\x, a_{13}, a_{14}$.
The second constraint trivially holds for any expectation.
\begin{flalign*}
% \intertext{For the while loop we return a fresh unknown expectation  and the following constraints. The symbolic state is restricted to $\x, a_{13}, a_{14}$, where $\x$ is the loop counter and $a_{13}, a_{14}$. The second constraint trivially hods for any expectation.}
    \phantom{xx} & \qinfer{\mbox{$WHILE{-}X$}}{\TCOST~0} = \GSA && \\
    & \phantom{x}\neg \x           \vdash \GSA \geq \underline{\qinfer{\pw{\CONSUME~1;\q_1 <* X;\q_1 <* H;\q_1,q_2 <* CNOT;\x^{\Bool} <- \MEAS{\q_1} }}{\GSA}} && \\
    & \phantom{x}\phantom{\neg} \x \vdash \GSA \geq \TCOST~0
\end{flalign*}
Next, the body is analysed. The symbolic state is not modified after the measurement.
\begin{flalign*}
    \phantom{xx}& \qinfer{\pw{\CONSUME~1;\q_1 <* X;\q_1 <* H;\q_1,\q_2 <* CNOT;\x^{\Bool} <- \MEAS{\q_1} }}{\GSA} \\
    % & = \qinfer{\pw{\CONSUME~1;\q_1 <* X;\q_1,\q_2 <* CNOT;\x^{\Bool} <- \MEAS{\q_1} }}{\GSA} &&\\
    & = \qinfer{\pw{\CONSUME~1;\q_1<* X;\q_1 <* H;\q_1,\q_2 <* CNOT}}{\TMEASURE~ (d_1\!+\!d_2, \GSB[0])~(d_3\!+\!d_4, \GSB[1])} &&\\
    &= \qinfer{\pw{\CONSUME~1}} {\TMEASURE~ (\PSYM{0}' , \GSB[0])~(\PSYM{0}', \GSB[1])} && \\
    &= 1 + {\TMEASURE~ ( \PSYM{0}', \GSB[0])~(\PSYM{1}', \GSB[1])} && 
\end{flalign*}
Therefore, the following constraint is constructed.
    \begin{flalign*}
        \phantom{xx}& \neg \x \vdash \GSA \geq 1 + {\TMEASURE~ ( \PSYM{0}', \GSB[0])~(\PSYM{1}', \GSB[1])} &&\\
\intertext{Consider the template $\alpha(l) \triangleq \mathsf{c_{13}} \cdot \guarded{ \neg \x \wedge a_{13} \rgeq 0} \cdot a_{13} + \mathsf{c_{24}} \cdot \guarded{ \neg \x \wedge a_{24} \rgeq 0} \cdot a_{24} + \mathsf{c} \cdot \guarded{\neg \x} \cdot 1$ with indeterminate variables $\mathsf{c_{13}}, \mathsf{c_{24}}, \mathsf{c}$.}
\phantom{xx}& \neg \x \vdash \alpha(\ell)(\astin) \geq 1 + {\TMEASURE~ ( \PSYM{0}', \alpha(\ell)(\symbolic{s_0}'))~(\PSYM{1}', \alpha(\ell)(\symbolic{\sigma_1}'))} && \\
\intertext{Applying the \textsc{Meas}, and the \textsc{FunR} rule we obtain.}
\phantom{xx} & \neg \x \wedge \isProb{\PSYM{0}'} \wedge \isProb{\PSYM{1}'} \vdash \alpha(\ell)(\astin) \geq 1 +  \PSYM{0}' \cdot (\guarded{\neg \x} \cdot \mathsf{c_{13}}) + \PSYM{1}' \cdot (\guarded{\neg \x} \cdot \mathsf{c_{24}}) && \\
&   \neg \x \wedge \isProbOne{\PSYM{0}'} \vdash \alpha(\ell)(\astin) \geq 1 + \PSYM{0}' \cdot (\guarded{\neg \x} \cdot \mathsf{c_{13}}) && \\
&   \neg \x \wedge \isProbOne{\PSYM{1}'} \vdash \alpha(\ell)(\astin) \geq 1 + \PSYM{1}' \cdot (\guarded{\neg \x} \cdot \mathsf{c_{24}})
    \end{flalign*}
These constraints can be further refined such that they are amenable to the \textsf{SMT} solver.
The assignment $\mathsf{c_{13}} = \mathsf{c_{24}} = \mathsf{c} = 2$ is a model.
Therefore, $g_\ell \triangleq 2 \cdot \guarded{ \neg \x \wedge a_{13} \rgeq 0} \cdot a_{13} + 2 \cdot \guarded{ \neg \x \wedge a_{24} \rgeq 0} \cdot a_{24} + 2 \cdot \guarded{\neg \x} \cdot 1$ is an upper bound.
We resume the inference.
   \begin{flalign*} 
        % \phantom{xx} & \qinfer{-X}{\TCOST~0} && \\
        % & \phantom{x}= \qinfer{\pw{\q_1 <* H;\q_1,q_2 <* CNOT;\x^{\Bool} <- \MEAS{\q_1};\mbox{$WHILE{-}X$}}}{\TCOST~0} && \\
        % & \phantom{x}= \qinfer{\pw{\q_1 <* H;\q_1,q_2 <* CNOT;\x^{\Bool} <- \MEAS{\q_1}}}{\qinfer{\mbox{$WHILE{-}X$}}{\TCOST~0}} && \\
        \phantom{xx} & \qinfer{\pw{\q_1 <* H;\q_1,q_2 <* CNOT;\x^{\Bool} <- \MEAS{\q_1}}}{\qinfer{\mbox{$WHILE{-}X$}}{\TCOST~0}} && \\
        & \phantom{x}= \qinfer{\pw{\q_1 <* H;\q_1,q_2 <* CNOT;\x^{\Bool} <- \MEAS{\q_1}}}{\TCOST~g_\ell} && \\
        & \phantom{x}= \qinfer{\pw{\q_1 <* H;\q_1,q_2 <* CNOT}}{\TCOST~(d_1 + d_2) \cdot 2 \cdot 1} && \\
        & \phantom{x}= \TCOST~(\sfrac{1}{2} + a_{13} + a_{24}) \cdot 2 \cdot 1 &&
    \end{flalign*}
For the parametrised state $\alpha\ket{10} + \beta\ket{11}$ we can derive the expected cost $1$.

% \subsection{Details}
%
% \begin{example}
%     Consider the post-measurement state of the projective measurement operator $M_{0,1}$ of a two-qubit system to perform a measurement that produces the classical outcome $0$ for qubit $q_1$.
% \begin{gather*}
%     \symbolic{\rho_0} = 
%     \begin{pmatrix}
%         (d_1,0)           & (a_{12}, b_{12}) & (a_{13},b_{13}) & (a_{14}, b_{14}) \\
%         (a_{12},-b_{12}) & (d_2,0)         & (a_{23},b_{23}) & (a_{24}, b_{24}) \\
%         (a_{13},-b_{13}) & (a_{23},-b_{23}) & (d_3,0)         & (a_{34}, b_{34}) \\
%         (a_{14},-b_{14}) & (a_{24},-b_{24}) & (a_{34},-b_{34}) & (d_4,0)
%     \end{pmatrix} \quad
% M_{0,1} =
% \begin{pmatrix}
% (1,0) & (0,0) & (0,0) & (0,0) \\
% (0,0) & (1,0) & (0,0) & (0,0) \\
% (0,0) & (0,0) & (0,0) & (0,0) \\
% (0,0) & (0,0) & (0,0) & (0,0)     
% \end{pmatrix}\\
% p_{0,1}(\symbolic{\rho_0}) = tr(M_{0,1} \symbolic{\rho_0}) = d_1 + d_2 \qquad
% m_{0,1}(\symbolic{\rho_0}) = \begin{pmatrix}
%     (\frac{d_1}{d_1+d_2}   ,0) & (\frac{a_{12}}{d_1+d_2},0) & (0,0) & (0,0) \\
%     (\frac{a_{12}}{d_1+d_2},0) & (\frac{d_2}{d_1+d_2},0) & (0,0) & (0,0) \\
%     (0,0) & (0,0) & (0,0) & (0,0) \\
%     (0,0) & (0,0) & (0,0) & (0,0)
% \end{pmatrix}
% \end{gather*}
% \end{example}

\subsection{Term Constraint Simplifications}
\label{app:term-simp}

We apply several simplification rules to reduce the complexity of the generated constraints.
The inference rules \textsc{[Cond-Top]} and \textsc{[Cond-Bot]} resolve case distinction in conditional expressions when the truth value of the guard can be inferred from the context.
In practice, this often resolves branching which is due to measurements and conditional branching as it occurs for example for qubit initialisations:
$\q^{\Qubits}  <- \ket{0} \triangleq \x <- \MEAS{\q} ; \IF \x \THEN \q <* \oper{X}  \ELSE \SKIP$.
In this case, measurements induce branching, however the case distinction for the conditional branching can be resolved using the simplification rules.
The inference rule \textsc{[MeasEq]} simplifies constraints in which the continuations are (syntactically) equivalent.

\begin{figure*}[ht]
    % \hrulefill
    \[
\begin{array}{c}
\infer[\textsc{[Cond-Top]}]{ \phi \vdash e \geq \TCOND~\bexp~\texp_1~\texp_2 }
    { \phi \implies \bexp \quad \phi \vdash e \geq \texp_1 }
\quad
\infer[\textsc{[Cond-Bot]}]{ \phi \vdash e \geq \TCOND~\bexp~\texp_1~\texp_2 }
    { \phi \implies \neg\bexp \quad \phi \vdash e \geq \texp_2 }
\\[2ex]
\infer[\textsc{[MeasEq]}]{ \phi \vdash e \geq \TMEASURE~(\aexp_1, \texp_1)~(\aexp_2, \texp_2) }
    {t_1 = t_2 \quad \phi \vdash e \geq t_1}
\end{array}
    \]
    % \vspace{-2mm}
    % \hrulefill
%    \vspace{-2mm}
    \vspace{-\belowdisplayskip}
\caption{Selection of inference rules for term constraints simplifications.}
\label{fig:term-simp}
\end{figure*}

\subsection{Algorithm for Expected Cost Inference}
\label{app:algorithm}

Algorithm~\ref{fig:alg} depicts the high-level algorithm for the expected cost inference.
It implements a bottom-up strategy to infer upper invariants for (nested) loops one after another and propagating the inferred cost expression.
Let $\cmd$ be a program. 
\Main{$\cmd$} starts the inference with the initial cost expression $\TCOST~0$ calling $\Solver(\cmd, \TCOST~0)$.
The function \Solver works on the top-level scope of the program.
It calls and propagates the symbolic inference $\qinfer{\cmd}{f}$ for non loop statements.
When $\cmd$ is a while statement then $\SolverGen$ is called to generate constraints for upper invariants.
A $\Stack$ object is used to keep track of constraints.
$\SolverGen$ first constructs a cost expression $g_\ell \triangleq \funzero$ for some fresh location symbol $\ell$.
It proceeds to construct the constraints for the upper invariant recursively.
For nested loops, multiple constraints are generated reflecting the control flow structure.
\SolverGen returns the cost expression $g_\ell$ which is an unknown expectation representing the upper bound for the loop statement plus the continuation which has been passed as argument to the function call.
This returns the control back to $\Solver$ which then tries to infer a certificate cost expression for $g_\ell$.
First, all generated constraints are popped from the stack.
Then, templates are generated from the cost expressions $f$ and $g_\ell$, as well as all generated constraints.
The function \ToPoly reduces the constraints to polynomial inequalities.
These constraints are then passed to an \textsf{SMT} solver.
If successful, the model is used to generate a bound for the cost expression $g_\ell$.
The inference may fail, in which \texttt{Unknown} is returned.

\begin{algorithm}[t]\small
\begin{minipage}{0.49\linewidth}%
\Fn{\Solver{$\cmd, f$}}{
\Switch{$\cmd$}{
    \Case{$\WHILE \bexp \DO \cmd_\ell$}{
        $g_\ell \gets \SolverGen{$\StackEmpty,\cmd_\ell, f$}$\\
        $cs_1 \gets \StackPopAll{}$ \\
        $t \gets \TemplateGen{$f, g_\ell, cs_1$}$\\
        $cs_2 \gets \ToPoly{$t,cs_1$}$\\
        $(sat,model) \gets \ToSmt{$cs_2$}$\\
        \If{$sat$}{
            $g \gets \ToBound{$g_\ell, model$}$\\
            \Return{$\TCOST~g$}
        }
        \Else{
            \Return{\texttt{Unknown}}
        }
    }
    \Case{$\cmd_1;\cmd_2$}{
        $f' \gets$ \Solver{$\cmd_2, f$}\\
        \Return{\Solver{$\cmd_1, f'$}}
    }
    \Else{
        \Return{$\qinfer{\cmd}{f}$};
    }
}
}
\end{minipage}%
\begin{minipage}{0.48\linewidth}
\Fn{\SolverGen{$\Stack, \cmd, f$}}{
\Switch{$\cmd$}{
    \Case{$\WHILE \bexp \DO \cmd_\ell$}{
        $g_\ell \gets \TFUN~\funsym~\symbolic{\astate_0}$\\
        $g_\ell' \gets$ \SolverGen{$\Stack, \cmd_\ell, g_\ell$}\\
        \Emit{$\phantom{\neg}\bexp \vdash g_\ell \geq g_\ell'$}\\
        \Emit{$\neg \bexp \vdash g_\ell \geq f$}\\
        \Return{$g_\ell$}\;
    }
    \Case{$\cmd_1;\cmd_2$}{
        $f' \gets$ \SolverGen{$\Stack, \cmd_2, f$}\\
        \Return{\SolverGen{$\Stack, \cmd_1, f'$}}
    }
    \Else{
        \Return{$\qinfer{\cmd}{f}$};
    }
    }
}
\vspace{1em}
\Fn{\Main{$\cmd$}}{
    \Return{\Solver{$\cmd, \TCOST~0$}}
}
\end{minipage}
\vspace{.5\baselineskip}
\caption{Bottom-up algorithm for expected cost inference.}
\label{fig:alg}
\end{algorithm}

Algorithm~\ref{fig:alg} outlines a bottom-up approach for inferring upper invariants and propagating cost expressions for loops.
\Main{$\cmd$} begins with an initial cost of zero, calling $\Solver(\cmd,\TCOST~0)$. 
The $\Solver$ function handles top-level statements, applies $\qinfer{\cmd}{f}$ for non-loop statements and $\SolverGen$ for while loops to generate constraints for upper invariants.
A stack object $\Stack$ is used to collect constraints of nested loops structures.
$\SolverGen$ constructs a cost expression $g_\ell$ with a fresh function symbol $\ell$, recursively generates constraints for nested loops, and returns $g_\ell$ as the upper bound for the loop plus continuation. 
$\Solver$ then attempts to infer a certificate for $g_\ell$ using templates to reduce collected constraints to polynomial inequalities via $\ToPoly$, and solving them with an \textsf{SMT} solver. 
If successful, a bound is derived and propagated to the outer context; otherwise, \texttt{Unknown} is returned.

\clearpage
\subsection{Benchmark Examples}
\label{app:case-studies}

Here, we illustrate selected examples from the benchmark.

\begin{figure}[ht] 
  \centering
    \begin{minipage}[t]{0.45\textwidth} 
        $COIN{-}TOSS(\q^{\Qubits}) \triangleq \\ 
        \phantom{xx}\pw{
        \x^{\Bool} = \true;\\
      \WHILE \x \DO {
        \CONSUME~1;\\
        \q <* H;\\
        \x = \MEAS{\q};\\
      }
        }
    $
    \end{minipage}
\caption{The quantum coin tossing program.}
\label{fig:coin-tossing}
\end{figure}

\newcommand{\anc}{\mathsf{anc}}

\newcommand{\tvar}{\mathsf{t}}
\newcommand{\kvar}{\mathsf{k}}

\begin{figure}[ht] 
    \centering 
    \begin{minipage}[c]{0.45\textwidth} 
        $FUSE(q3,q4,x2) \triangleq \\
        \phantom{xx}\pw{
          \CONSUME~1;\\
          \anc_1^{\Qubits} = \ket{+};
          \anc_2^{\Qubits} = \ket{+};\\
          \cq_1 <- \MEAS{\anc_1};\\
          \cq_2 <- \MEAS{\anc_2};\\
          \IF \cq_1 \wedge \cq_2 \THEN {
            \q_3,\q_4 <* CZ;\\
            \x_2 = \true;\\
          } \ELSE {
            \x_2 = \MEAS{\q_3};\\
            \x_2 = \MEAS{\q_4};\\
            \x_2 = \false;\\
          }
          }$
\\[5pt]
        $CHAIN4() \triangleq \\
        \phantom{xx}\pw{
        \x_0 ^{\Bool} = \false;\\
        \WHILE \neg \x_0 \DO {
          \q_1 = \ket{+};
          \q_2 = \ket{+};\\
          FUSE(\q_1,\q_2,\x_1);
        }\\

        \x_2 ^{\Bool} = \false;\\
        \WHILE \neg \x_2 \DO {
          \q_3 = \ket{+};
          \q_4 = \ket{+};\\

          FUSE(\q_3,\q_4,\x_2);
        }\\
        FUSE(\q_2,\q_3,\x_0);
        }
    $
    \end{minipage}
    \
    \begin{minipage}[c]{0.45\textwidth}
        $CHAIN() \triangleq \\
        \phantom{xx}\pw{
        \tvar = 0;\\
        \WHILE \kvar > \tvar \wedge \tvar >= 0 \DO {
          \CMT{$CHAIN(\q_{t+1},\q_{t+2},q_{t+3},\q_{t+4})$}\\
          CHAIN4();\\
          \CMT{$FUSE(\q_{t+1},\q_{t+2},\x)$}\\
          \CONSUME~1;\\
          \anc_1^{\Qubits} = \ket{+};\\
          \anc_2^{\Qubits} = \ket{+};\\
          \cq_1 <- \MEAS{\anc_1};\\
          \cq_2 <- \MEAS{\anc_2};\\
          \IF \cq_1 \wedge \cq_2 \THEN {
            \tvar = \tvar + 4;\\
          } \ELSE {
            \tvar = \tvar - 1;\\
          }
        }
        }
        $
    \end{minipage}
\caption{A quantum program to create a chain of $k$ entangled qubits.}
\label{fig:fuse-chain}
\end{figure}

\begin{figure}[ht] 
    \centering 
    \begin{minipage}[t]{0.90\textwidth} 
        $(4I+iZ)/\sqrt17 \triangleq \\
        \phantom{xx}\pw{
          \x = \true;\\
          \WHILE \x \DO {
            \CONSUME~1;\\

    \q = |0>;
    \q <* X; \\

    \q <* H;
    \q <* S; \\
    \q <* T; 
    \q <* H; 
    \q <* T; 
    \q <* H; 
    \q <* T; 
    \q <* H; 
    \q <* T; 
    \q <* S^\dagger; 
    \q <* H;  \\

    \q,q' <* CZ; \\

    \q' <* X;

    \q <* H;
    \q <* S;
    \q <* H; \\
    \q <* T; 
    \q <* H; 
    \q <* T; 
    \q <* H; 
    \q <* T; 
    \q <* H;  \\

    \q,\q' <* CZ; \\

    \q' <* X;

    \q <* H;
    \q <* S;
    \q <* H;
    \q <* T; \\
    \q <* H;
    \q <* T;
    \q <* H; 
    \q <* T; 
    \q <* H; 
    \q <* T; 
    \q <* H; 
    \q <* S; 
    \q <* H; \\

    \x <- \MEAS{q};\\
        }
        }
    $
    \end{minipage}
\caption{RUS program implementing $(4I+iZ)/\sqrt17$.}
\label{fig:rus-10b}
\end{figure}

\clearpage
\subsection{Omitted Proofs}
\label{ss:proofs}

\begin{theorem}[Soundness Theorem]
  Let $\cmd$ be a program, $f$ be a cost expression, and $\alpha$ be an assignment.
  Then $\qet{\cmd}{\sem[\alpha]{f}} \leq \sem[\alpha]{\qinfer{\cmd}{f}}$, if all side-conditions are met.
\end{theorem}

\begin{proof}
  \label{p:soundness}
  The proof is by induction on the structure of the program $\cmd$.

  \textbf{Case}$[\SKIP]$: Then by unfolding.
  \begin{align*}
    \qet{\SKIP}{\sem[\alpha]{f}} = \sem[\alpha]{f} = \sem[\alpha]{\qinfer{\SKIP}{f}}
  \end{align*}

  \textbf{Case}$[\TCONSUME~n]$: Then by unfolding and case distinction on non-negativity of the tick expression.
  \begin{align*}
  & \qet{\TCONSUME~\nexp}{\sem[\alpha]{f}} \\
  & = \max(0,\sem[\alpha]{\nexp}) + \sem[\alpha]{f}\\
  & = \sem[\alpha]{\TCONSUME~\nexp~f} \\
  & =\sem[\alpha]{\qinfer{\TCONSUME~\nexp}{f}}
  \end{align*}

  \textbf{Case}$[\x <- \e]$: Then by unfolding and the definition of substitution.
  \begin{align*}
  & \qet{\x <- \e}{\sem[\alpha]{f}} \\
  & = \sem[\alpha]{f}[\x := \e]\\
  & = \sem[\alpha]{f\{\x \mapsto \e\}} \\
  & =\sem[\alpha]{\qinfer{\x <- \e}{f}}
  \end{align*}

  \textbf{Case}$[\qs <* \ope]$: Then by unfolding and the definition of substitution.
  \begin{align*}
  & \qet{\qs <* \ope }{\sem[\alpha]{f}} \\ 
  & = \sem[\alpha]{f}[\Phi_{U_{\qs}}] \\
  & = \sem[\alpha]{f\{\qs \mapsto U\symbolic{\rho_0}U^\dagger\}} \\
  & = \sem[\alpha]{\qinfer{\qs <* \ope }{f}}           
  \end{align*}

  \textbf{Case}$[\IF \bexp \THEN \cmd_1 \ELSE \cmd_2]$: Then by unfolding, case distinction on the guard and IH.
  \begin{align*}
  & \qet{\IF \bexp \THEN \cmd_1 \ELSE \cmd_2}{\sem[\alpha]{f}} \\ 
  & = \wpt{\cmd_1}{\sem[\alpha]{f}}\up{\sem[\alpha]{\bexp}} \wpt{\cmd_2}{\sem[\alpha]{f}} \\
  & \leq \sem[\alpha]{\qinfer{\cmd_1}{f}}\up{\sem[\alpha]{\bexp}} \sem[\alpha]{\qinfer{\cmd_2}{f}} \\
  & = \sem[\alpha]{\TCOND~\bexp~\qinfer{\cmd_1}{f}~\qinfer{\cmd_2}{f}} \\
  & = \sem[\alpha]{\qinfer{\IF \bexp \THEN \cmd_1 \ELSE \cmd_2}{f}}
  \end{align*}

  \textbf{Case}$[\cmd_1 ; \cmd_2]$: Then by unfolding and IH.
  \begin{align*}
      & \qet{\cmd_1 ; \cmd_2}{\sem[\alpha]{f}} \\ 
      & =\qet{\cmd_1}{\wpt{\cmd_2}{\sem[\alpha]{f}}} \\
      & \leq \qet{\cmd_1}{\sem[\alpha]{\qinfer{\cmd_2}{f}}} \\
      & \leq \sem[\alpha]{\qinfer{\cmd_1}{\qinfer{\cmd_2}{f}}}\\
      & = \sem[\alpha]{\qinfer{\cmd_1 ; \cmd_2}{f}}
  \end{align*}

  \textbf{Case}$[\x <- \MEAS{\q_i}]$: Then by unfolding, definition of substitution and case distinction on the probability.
  \begin{align*}
  & \wpt{\x <- \MEAS{\q_i}}{\sem[\alpha]{f}} \\
    & = \sem[\alpha]{f}[\x := 0; m_{0,i}] \up{\proba{0}{i}} \sem[\alpha]{f}[\x := 1; m_{1,i}] \\
  & = \sem[\alpha]{\TMEASURE~(p_{0,i}(\symbolic{\rho_0}), f\{\x \mapsto 0; \qs \mapsto m_{0,i}(\symbolic{\rho_0})\}) (p_{1,i}(\symbolic{\rho_0}),  f\{\x \mapsto 1; \qs \mapsto m_{1,i}(\symbolic{\rho_0})\})} \\
  & = \sem[\alpha]{\qinfer{\x <- \MEAS{\q_i}}{f}}
  \end{align*}

  \textbf{Case}$[\WHILE \bexp \DO \cmd]$: Then by unfolding, upper invariant law, and the assumption that all side-constraints hold.
  \begin{align*}
  & \wpt{\WHILE \bexp \DO \cmd}{\sem[\alpha]{f}} \\
  & \leq \alpha(\ell)(\symbolic{\astate_0}) \\
  & = \sem[\alpha]{\funzero} \\
  & = \sem[\alpha]{\qinfer{\WHILE \bexp \DO \cmd_\ell}{f}}
  \end{align*}
\end{proof}

% \begin{lemma}
%   Let $\phi \vdash e \geq f$ be a generated constraint from $\qinfer{\cdot}{\cdot}$ and let $\Delta$ be a derivation from $\phi \vdash e \geq f$ using the inference rules of Figure~\ref{fig:term-simp} and Figure~\ref{fig:term-to-cost}.    
%   Then $\phi \vdash e \geq f$ holds, if all derived inequalities in $\Delta$ hold.
% \end{lemma}
% \begin{proof}
%   \label{p:inference}
%   By induction on the term structure and case distinction. 
% \end{proof}

\end{document}